\documentclass[preprint,12pt]{elsarticle}

\usepackage{amsfonts}
\usepackage{amssymb}
\usepackage{amsmath}
\usepackage{amsthm}
\usepackage[letterpaper]{geometry}
\usepackage[parfill]{parskip}
\usepackage{color}
\usepackage{url}
\usepackage{graphicx}
\usepackage{xspace}
\usepackage{natbib}
\usepackage[usenames,dvipsnames]{xcolor}

\usepackage[ruled,vlined,english,linesnumbered]{algorithm2e}

\newcommand{\leo}{}
\newcommand{\cs}{}
\newcommand{\csN}{}

\newcommand{\revA}{} 
\newcommand{\revB}{} 
\newcommand{\revC}{} 
\newcommand{\revX}{} 

\newcommand{\revAnew}{}
\newcommand{\revBnew}{}

\newcommand{\Np}{\mathbb{N}^{+}}
\newcommand{\Dp}{\Delta^{+}}
\newcommand{\NN}{\mathbb{N}}
\newcommand{\cT}{\mathcal{T}}
\newcommand{\HN}{{\sc Hybridization Number}\xspace}
\newtheorem{observation}{Observation}
\newtheorem{claim}{Claim}
\newtheorem{lemma}{Lemma}
\newtheorem{definition}{Definition}
\newtheorem{theorem}{Theorem}
\newtheorem{corollary}{Corollary}

\journal{Journal of Computer and System Sciences}

\begin{document}

\begin{frontmatter}



\title{Kernelizations for the hybridization number problem\\ on multiple nonbinary trees\tnoteref{prelim}}

\tnotetext[prelim]{A preliminary version of this article appeared in the proceedings of Workshop on Graph-Theoretic Concepts \revB{in Computer Science} (WG 2014).}

\author[TUD]{Leo van Iersel\fnref{veni}}\ead{l.j.j.v.iersel@gmail.com}
\author[DKE]{Steven Kelk}\ead{steven.kelk@maastrichtuniversity.nl}
\author[ISEM]{Celine Scornavacca}\ead{celine.scornavacca@univ-montp2.fr}

\fntext[veni]{Leo van Iersel was \revX{partially} funded by a Veni grant from The Netherlands Organisation for Scientific Research (NWO).}

\address[TUD]{\revX{Delft Institute of Applied Mathematics, Delft University of Technology, P.O. Box 5, 2600 AA Delft, The Netherlands}}
\address[DKE]{Department of Knowledge Engineering (DKE), Maastricht University, P.O. Box 616, 6200 MD Maastricht, The Netherlands}
\address[ISEM]{ISEM, CNRS -- Universit\'e Montpellier II, Place Eug\`ene Bataillon, 34095, Montpellier, France }

\begin{abstract}
Given a finite set $X$, a collection $\cT$ of rooted phylogenetic trees on $X$ and an integer~$k$, the \HN problem asks if there exists a phylogenetic network on~$X$ that displays all trees from~$\cT$ and has reticulation number at most~$k$. We show two kernelization algorithms for \HN, with kernel sizes $4k(5k)^t$ and $20k^2(\Delta^+-1)$ respectively, with~$t$ the number of input trees and~$\Delta^+$ their maximum outdegree. \revX{Experiments on simulated data demonstrate the practical relevance of our kernelization algorithms.} In addition, we present an $\revX{n^{f(k)}t}$\revB{-}time algorithm, with $n=|X|$ and $f$ some computable function of~$k$.
\end{abstract}


\end{frontmatter}


\section{Introduction} 
In phylogenetics, one central challenge is to construct a plausible evolutionary history for a set of contemporary species~$X$ given incomplete data. This usually concerns biological evolution, but the paradigm is equally applicable to more abstract forms of evolution, e.g. natural languages~\cite{NakhlehLanguages}. Classically an evolutionary history is modelled by a \emph{rooted phylogenetic tree}, essentially a rooted tree in which the leaves are bijectively labelled by~$X$~\cite{SempleSteel2003}. In recent years, however, there has been growing interest in generalizing this model to directed acyclic graphs, \revB{that is,} to \emph{rooted phylogenetic networks} \cite{expanding,HusonRuppScornavacca10,davidbook}. In the latter model, \emph{reticulations}, \cs{which are vertices of indegree~2 or higher,} are of central importance; these are used to represent non-treelike evolutionary phenomena such as hybridization and lateral gene transfer. The \emph{reticulation number} of a phylogenetic network can be defined as the number of edges \cs{that need to be removed in order to obtain a tree}. It is easy to see that in networks with maximum indegree~2 (to which we will be able to restrict without loss of generality) the reticulation number is simply equal to the number of reticulations. This \leo{setting} has naturally given rise to the \textsc{Hybridization Number} problem: given a set of \leo{rooted} phylogenetic trees $\cT$ on the same set of taxa~$X$, construct a \leo{rooted} phylogenetic network on~$X$ with \cs{the} smallest possible reticulation number, such that an image of every tree in~$\cT$ is embedded in the network~\cite{BaroniEtAl2005}.

\HN has attracted considerable interest in a short space of time. Even in the case when~$\cT$ consists of two binary (\revB{that is,} bifurcating) trees the problem is $\textsf{NP}$-hard, $\textsf{APX}$-hard~\cite{bordewich07a} and in terms of approximability is a surprisingly close relative of the problem \textsc{Directed Feedback Vertex Set}~\cite{cyclekiller,nonbinary}. On the positive side, this variant of the problem is fixed-parameter tractable ($\textsf{FPT}$) in parameter~$k$, the \leo{reticulation number of an optimal network}. Initially this was established via kernelization~\cite{sempbordfpt2007}, but more recently efficient bounded-search algorithms have emerged with \leo{$O( 3.18^k \cdot \text{poly}(n))$ being the current state of the art~\cite{whidden2013fixed}, with~$n=|X|$}.

\clearpage

In this article we focus on the general case when~$t=|\cT| \geq 2$ and the trees
in~$\cT$ are \cs{not necessarily binary}. This causes complications for two reasons. First, when~$t > 2$, the popular ``maximum acyclic agreement forest'' abstraction breaks down, a central pillar of algorithms for the~$t=2$ case. Second, in the nonbinary case the images of the trees in the network are allowed to be more ``resolved'' than the original trees. (More formally, an input tree~$T$ is seen as being embedded in a network~$N$ if~$T$ can be obtained from a subgraph of~$N$ by contracting edges.) The reason for this is that vertices with outdegree greater than two are used by biologists to model uncertainty in the order that species diverged. Both factors complicate matters considerably. Consequently, progress has been more gradual.

For the case of multiple binary trees, there exists a kernel \revB{with at most $20k^2$ leaves}~\cite{vanIerselLinz}, various heuristics~\cite{chen2012algorithms,chen2013ultrafast,pirnISMB2010} and an exact approach without running-time bound~\cite{wu2013algorithm}.

For the case of two nonbinary trees, there is also a polynomial kernel~\cite{linzsemple2009}, based on a highly technical kernelization argument, and a simpler $\textsf{FPT}$ algorithm based on bounded search~\cite{simplefpt}.

This leaves the \leo{general} case of \leo{multiple} nonbinary trees as the main variant for which it is unclear whether the problem is $\textsf{FPT}$. \leo{The most obvious parameter choice is, as before, the reticulation number~$k$. However, other natural parameters in this case are the number of input trees~$t$ and the maximum outdegree~$\Dp$ over all input trees. By the NP-hardness result mentioned above, it is clear that \HN is not FPT if the parameter is~$t$ or~$\Dp$, unless $\textsf{P}$ $=$ $\textsf{NP}$. Therefore, the most interesting questions are whether the problem is FPT if either (a) the only parameter is~$k$, or (b) there are two parameters:~$k$ and~$t$, or (c) there are two parameters:~$k$ and~$\Dp$.}

In this paper, we answer the latter two questions \revX{affirmatively}, using a kernelization approach. First, we prove that \HN admits a kernel with at most~$4k(5k)^t$ leaves. Second, we show a slightly different kernel with at most~$20k^2(\Dp-1)$ leaves. The running time of both kernelization algorithms is polyomial in~$n$ and~$t$. Whether \HN remains FPT if~$k$ is the only parameter remains open. However, we do present an algorithm for \HN that runs in $\revX{n^{f(k)}t}$ time, with~$f$ some computable function of~$k$, hence showing that the problem is in the class $\textsf{XP}$.

\leo{Similar results can alternatively be obtained using bounded-search algorithms instead of kernelization, see the e-print~\cite{towards}. We do not include those algorithms here because the proofs (although based on several important insights) are highly technical and the running times astronomical. In contrast, the kernelization algorithms are simple, fast and their proofs relatively elegant. Therefore, we only include the last result of our e-print~\cite{towards}, which is the $\revX{O(n^{f(k)}t)}$ time algorithm, in this paper. Its running time is also astronomical but, combined with the kernelization algorithms, it gives explicit FPT algorithms, which are (theoretically) the best known algorithms that can solve general instances of \HN.}

\leo{Some of the lemmas that we prove in order to derive the correctness of the kernelization algorithms are of \revA{independent} interest because they improve our understanding of how nonbinary trees can be embedded inside networks. This helps us to avoid a technical case analysis (as in~\cite{linzsemple2009}) and exhaustive guessing (as in~\cite{towards}), leading to a simple and unified kernelization approach that is applicable to a more general problem (compared to e.g.~\cite{linzsemple2009,vanIerselLinz}).}

\leo{Moreover, the $4k(5k)^t$ kernel introduces an interesting way to deal with multiple parameters simultaneously. It is based on searching, for decreasing~$q$, for certain substructures called ``$q$-star chains", which are chains that are common to all~$t$ input trees and form stars in~$q$ of the input trees. When we encounter such substructures we truncate them to a size that is a function of~$q$ and~$k$. Since we loop through all possible values of~$q$ \cs{($0\leq q\leq t$)}, we eventually truncate all common substructures. The correctness of each step heavily relies on the fact that substructures for larger values of~$q$ have already been truncated. However, when~$q$ decreases, the size to which substructures can be reduced increases \cs{(as will become clear later)}. This has the effect that the size of kernelized instances is a function of~$k$ and~$t$ \cs{and not of~$k$ only}. For the $20k^2(\Dp-1)$ kernel, we use a similar but simpler technique.}

\leo{From our results it follows that \HN admits a polynomial-size kernel in the case that either the number of input trees or their maximum outdegrees are bounded by a constant. Moreover, the kernelization algorithms run in polynomial time for general instances, with an unbounded number of trees with unbounded outdegrees. The main remaining open problem is to determine whether \HN remains fixed-parameter tractable if the input consists of an unbounded number of trees with unbounded outdegrees and the only parameter is the reticulation number~$k$.}

\revX{Finally, to demonstrate the practical relevance of the kernelization algorithms presented in this article, we have implemented them in Java and studied their performance under a variety of experimental parameters. Our experiments show that for large trees (500-1000 taxa) the kernelizations run quickly 
and in many cases a reduction in instance size of 90$\%$ or more is achieved. The experiments also yield insight into the conditions under which the different kernelization algorithms do and do not effectively reduce the size of instances. The code, which combines all the kernelization algorithms into a single package, is
freely available at\\ \url{http://leovaniersel.wordpress.com/software/treeduce/}.}

\section{Preliminaries}
Let~$X$ be a finite set. A \emph{rooted phylogenetic} $X$-\emph{tree} is a rooted tree with no vertices with indegree~1 and outdegree~1, a root with indegree~0 and outdegree at least~2, and leaves bijectively labelled by the elements of~$X$. We identify each leaf with its label. We henceforth call a rooted phylogenetic $X$-tree a \emph{tree} (\emph{on}~\revAnew{$X$}) for short. A tree~$T$ is a \emph{refinement} of a tree~$T'$ if~$T'$ can be obtained from~$T$ by contracting edges.

\begin{figure}[t]
    \centering
    \includegraphics[scale=.35]{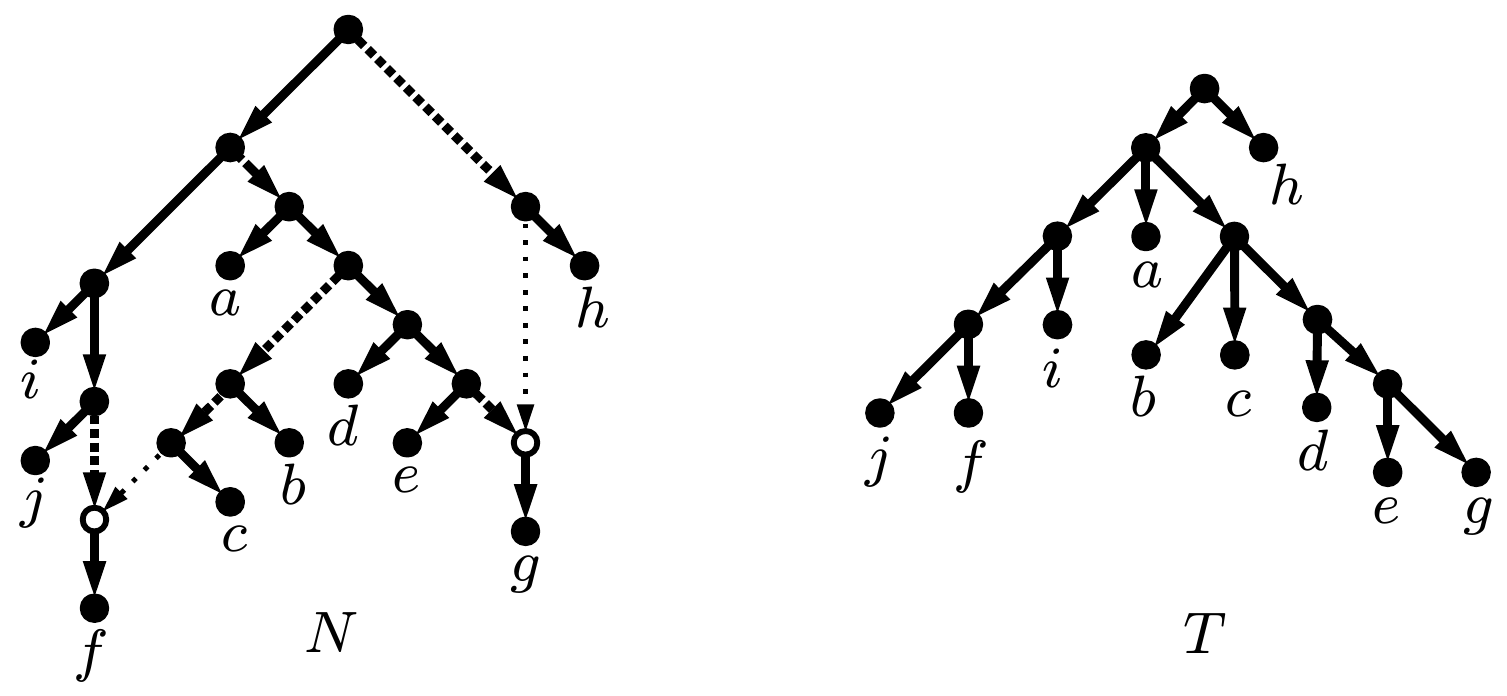}
    \caption{A (rooted phylogenetic) network~$N$ and a (rooted phylogenetic) tree~$T$. Network~$N$ is binary, has two reticulations (unfilled) and reticulation number~2. Tree~$T$ is displayed by~$N$ because it can be obtained from~$N$ by deleting the \cs{dotted} edges and contracting the \cs{dashed} edges.\label{fig:intro}}
\end{figure}

Throughout the paper, we refer to directed edges simply as edges. If~$e=(u,v)$ is an edge, then we say that~$v$ is a \emph{child} of~$u$, that~$u$ is a \emph{parent} of~$v$, that~$v$ is the \emph{head} of~$e$ \leo{and that~$u$ is the \emph{tail} of~$e$.}

A \emph{rooted phylogenetic network} (on~$X$) is a directed acyclic graph with no vertices with indegree~1 and outdegree~1\cs{, a single indegree-0 vertex (the \emph{root})}, and leaves \revA{(vertices with outdegree~0)} bijectively labelled by the elements of~$X$. Rooted phylogenetic networks will henceforth be called \emph{networks} for short in this paper. A tree~$T$ is \emph{displayed} by a network~$N$ if~$T$ can be obtained from a subgraph of~$N$ by contracting edges. See Figure~\ref{fig:intro} for an example. \leo{Note that, without loss of generality, we may assume that edges incident to leaves are not contracted.} Using~$d^-(v)$ to denote the indegree of a vertex~$v$, a \emph{reticulation} is a vertex~$v$ with~$d^-(v)\geq 2$. The \emph{reticulation number} of a network~$N$ with vertex set~$V$ \leo{and edge set~$E$ is defined as $r(N)=|E|-|V|+1$ or, equivalently, as}

\[
r(N)=\sum_{v\in V : d^-(v)\geq 2}(d^-(v)-1).
\]

Given a set of trees~$\cT$ on~$X$, we use $r(\cT)$ to denote the minimum value of~$r(N)$ over all \revA{networks}~$N$ on~$X$ that display~$\cT$. We are now ready to formally define the problem we consider.

\noindent{\bf Problem:} \HN\\
\noindent {\bf Instance:} A finite set~$X$, a collection~$\cT$ of \revA{trees} on~$X$ and~$k\in\Np$. \\
\noindent {\bf Question:} Is $r(\cT)\leq k$, \revB{that is}, does there exist a \revA{network}~$N$ on~$X$ that displays~$\cT$ and has~$r(N)\leq k$.

A network is called \emph{binary} if each vertex has indegree and outdegree at most~2 and if each vertex with indegree~2 has outdegree~1. By the following \leo{observation} we may restrict to binary networks.
\medskip
\begin{observation}[\cite{towards}]\label{obs:binary}
If there exists a network~$N$ on~$X$ that displays~$\cT$ then there exists a binary network~$N'$ on~$X$ that displays~$\cT$ such that~$r(N)=r(N')$.
\end{observation}

The observation follows directly from noting that, for each network~$N$, there exists a binary network~$N'$ with~$r(N')=r(N)$ such that~$N$ can be obtained from~$N'$ by contracting edges. Hence, any tree displayed by~$N$ is also displayed by~$N'$.

\revA{A subgraph~$T'$ of a network~$N$ (which may be a tree) is said to be a \emph{pendant subtree} if \csN{$T'$ does not contain reticulations and if} there is no non-root vertex of~$T'$ \revAnew{that} has a child or parent in~$N$ that is not in~$T'$. Note that a pendant subtree of a network on~$X$ is a tree on~$Y$ with~$Y\subseteq X$ and~$Y\neq\emptyset$. If~$|Y|=1$ then the subtree is called \emph{trivial}.} 

\revA{We use~$p_N(v)$ to denote the \csN{set of parents}
of a vertex~$v$ in a network~$N$ (which may be a tree). If~$x$ and~$y$ are leaves of~$N$, then we say that~$x$ is \emph{above}~$y$ in~$N$ if~$N$ contains a directed path from \csN{a node in}~$p_N(x)$ to~$y$. If, in addition, $p_N(x)\neq p_N(y)$, we say that~$x$ is \emph{strictly above}~$y$. Observe that two leaves~$x$ and~$y$ have a common parent in~$N$ if and only if~$x$ is above~$y$ and~$y$ is above~$x$.}

\revB{\emph{Suppressing} a vertex~$v$ with indegree~1 and outdegree~1 means adding an edge from the parent of~$v$ to the child of~$v$ and subsequently deleting~$v$.}

\begin{figure}[t]
    \centering
    \includegraphics[scale=.35]{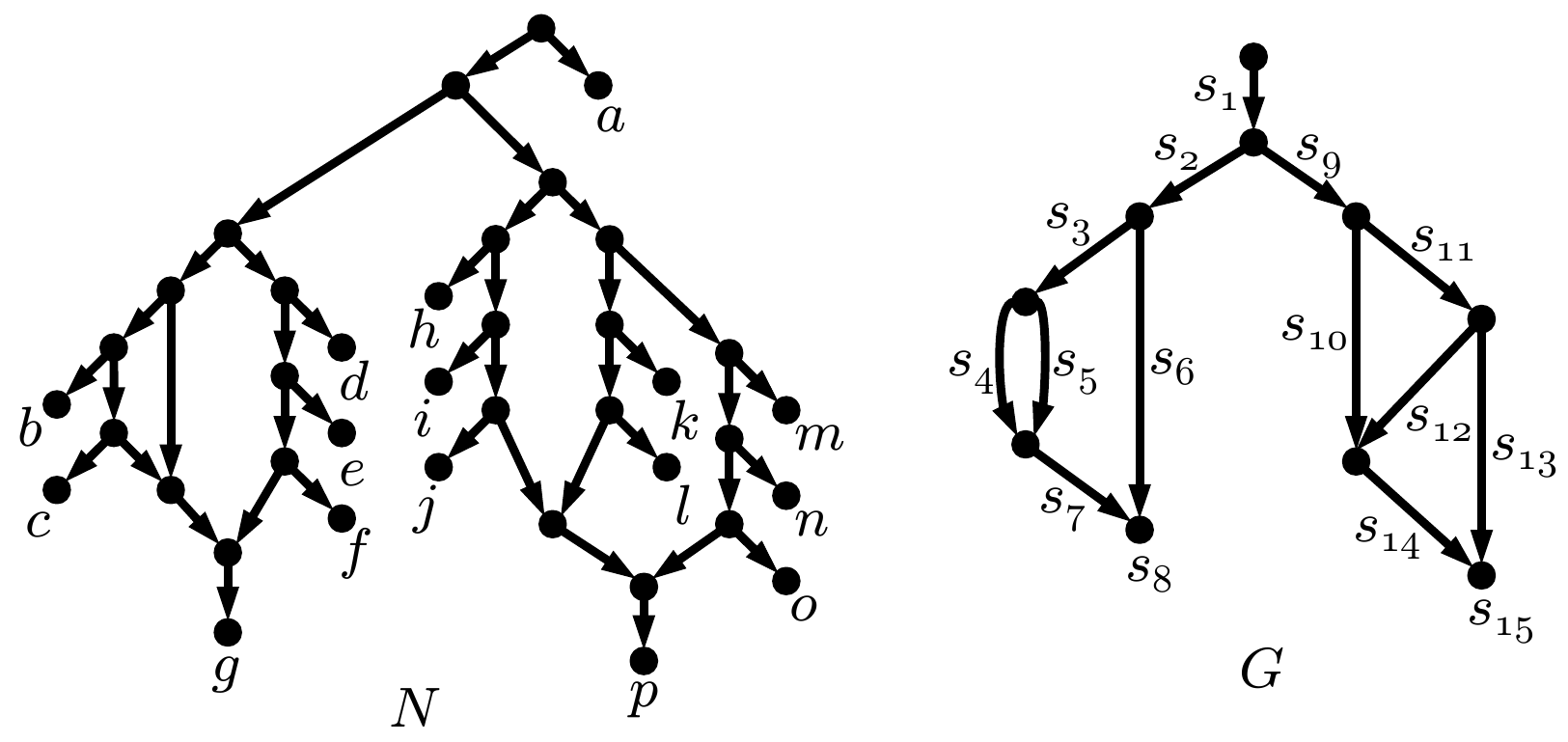}
    \caption{A network~$N$ and the 4-reticulation generator~$G$ underlying~$N$. Generator~$G$ has two vertex sides~$s_8$ and~$s_{15}$ and 13 edge sides. For example, leaves~$d,e$ and~$f$ are on edge side~$s_6$ and leaf~$g$ is on vertex side~$s_8$.\label{fig:generator}}
\end{figure}

The notion of ``generators'' is used to describe the underlying structure of a network without nontrivial pendant subtrees~\cite{KelkScornavacca2011}. Let~$k\in\Np$. A \emph{binary $k$-reticulation generator} is defined as an acyclic directed multigraph with a single root with indegree~0 and outdegree~1, exactly~$k$ vertices with indegree~2 and outdegree at most~1, and all other vertices have indegree~1 and outdegree~2. See Figure~\ref{fig:generator} for an example. Let~$N$ be a binary network with no nontrivial pendant subtrees and with~$r(N)=k$. Then, a binary $k$-reticulation generator is said to be the \emph{generator underlying}~$N$ if it can be obtained from~$N$ by adding a new root with an edge to the old root, deleting all leaves and suppressing all resulting indegree-1 outdegree-1 vertices. In the other direction,~$N$ can be reconstructed from its underlying generator by subdividing edges, adjoining a leaf to each vertex that subdivides an edge, or has indegree~2 and outdegree~0, via a new edge, and deleting the outdegree-1 root. The \emph{sides} of a generator are its edges (the \emph{edge sides}) and its vertices with indegree~2 and outdegree~0 (the \emph{vertex sides}). Thus, each leaf of~$N$ is on a certain side of its underlying generator. To formalize this, consider a leaf~$x$ of a binary network~$N$ without nontrivial pendant subtrees and with underlying generator~$G$. If the parent~$p$ of~$x$ has indegree~2, then~$p$ is a vertex side of~$G$ and we say that~$x$ \emph{is on side}~$p$. If, on the other hand, the parent~$p$ of~$x$ has indegree~1 and outdegree~2, then~$p$ is used to subdivide an edge side~$e$ of~$G$ and we say that~$x$ \emph{is on side}~$e$. We say that two \revA{leaves~$x$ and~$y$ are} \emph{on the same side} of~$N$ if the underlying generator of~$N$ has an edge side~$e$ such that~$x$ and~$y$ are both on side~$e$. The following lemma will be useful.
\medskip
\begin{lemma}[\cite{vanIerselLinz}]\label{lem:sides}
If~$N$ is a binary \revA{network} with no nontrivial pendant subtrees and with $r(N)=k>0$ and if~$G$ is its underlying generator, then~$G$ has at most~$4k-1$ edge sides, at most~$k$ vertex sides and at most~$5k-1$ sides in total.
\end{lemma}

A \emph{kernelization} of a parameterized problem is a polynomial-time algorithm that maps an instance~\cs{$I$} with parameter~$k$ to an instance~\cs{$I'$} with parameter~$k'$ such that (1)~$(I',k')$ is a yes-instance if and only if~$(I,k)$ is a yes-instance, (2) the size of~$I'$ is bounded by a function~$f$ of~$k$, and (3) the size of~$k'$ is bounded by a function of~$k$~\cite{downey1999}. A kernelization is usually referred to as a \emph{kernel} and the function~$f$ as the \emph{size} of the kernel. Thus, a parameterized problem admits a polynomial kernel if there exists a kernelization with~$f$ being a polynomial. A parameterized problem is \emph{fixed-parameter tractable} ($\textsf{FPT}$) if there exists an algorithm that solves the problem in time $O(g(k)|I|^{O(1)})$, with~$g$ being some \revX{computable} function of~$k$ and~$|I|$ the size of~$I$. It is well known that a parameterized problem is fixed-parameter tractable if and only if it admits a kernelization and is decidable. However, there exist fixed-parameter tractable problems that do not admit a kernel of \emph{polynomial} size unless the polynomial hierarchy collapses~\cite{bodlaender}. Kernels are of practical interest because they can be used as polynomial-time preprocessing which can be combined with any algorithm solving the problem (\leo{e.g. an exponential-time exact algorithm or a heuristic}). \revB{The class $\textsf{XP}$ contains all parameterized problems that can be solved in $n^{h(k)}$ time, with~$h$ a computable function of the parameter~$k$.} 

\section{A polynomial kernel for a bounded number of trees}\label{sec:boundedtrees}

We first introduce the following key definitions. Let~$\cT$ be a set of trees. A tree~\cs{$S$} is said to be a \emph{common pendant subtree} of~$\cT$ if it is a refinement of a pendant subtree of each~$T\in\cT$ and~$S$ is said to be \emph{nontrivial} if it has at least two leaves.

\medskip

\revB{The kernelization is described in Algorithm~\ref{alg:kernel}. We will give the definition of (common $q$-star) chains after Lemmas~\ref{lem:subtreereduction} and~\ref{lem:polysubtree}, which show that the subtree reduction} preserves the reticulation number and can be applied in polynomial time. \revB{Their proofs} use the following definition. \emph{Cleaning up} a directed graph means repeatedly deleting unlabelled outdegree-0 vertices and indegree-0 out\-degree-1 vertices, suppressing indegree-1 out\-degree-1 vertices and replacing multiple edges by single edges until none of these operations is applicable. Cleaning up is used to turn directed graphs into valid networks and it can easily be checked that the cleaning-up operation does not affect which trees are being displayed \revB{and its result does not depend on the order of the operations}.
\medskip
\begin{lemma}\label{lem:subtreereduction}
Let $(X,\cT,k)$ be an instance of \HN and let $(X',\cT',k)$ be the instance obtained after applying the subtree reduction for \cs{a} common pendant subtree~$S$. \revB{Then}~$r(\cT)\leq k$ if and only if~$r(\cT')\leq k$.
\end{lemma}
\begin{proof}
If~$r(\cT')\leq k$ then clearly also~$r(\cT)\leq k$ because in any network~$N'$ displaying~$\cT'$ we can simply replace leaf~$x$ by \cs{the} pendant subtree~$S$ to obtain a network~$N$ that displays~$\cT$ and has~$r(N)=r(N')$.

Now suppose that~$r(\cT)\leq k$, \revB{that is}, that there exists a \revA{network}~$N$ that displays~$\cT$ and has~$r(N)\leq k$. We construct a network~$N'$ displaying~$\cT'$ from~$N$ in the following way. Pick any leaf~$y$ of~$S$. Delete all leaves of~$S$ except for~$y$ from~$N$ and relabel~$y$ to~\revA{$x^\dagger$}. Let~$N'$ be the result of cleaning up the resulting digraph. It is easy to check that~$N'$ displays~$\cT'$ and that~$r(N')\leq r(N)$. Hence,~$r(\cT')\leq k$.
\end{proof}
\medskip
\begin{lemma}\label{lem:polysubtree}
Given a set~$\cT$ of trees on~$X$, there exists \leo{an $O(|X|^3|\cT|)$ time} algorithm that decides if there exists a nontrivial maximal common pendant subtree of~$\cT$ and constructs such a subtree if it exists.
\end{lemma}
\begin{proof}
If there exist no two leaves that have a common parent in each tree in~$\cT$, then there are no nontrivial common pendant subtrees and we are done.

\begin{algorithm}[t]
{\textbf{Subtree Reduction:}
\If{there is a nontrivial maximal common pendant subtree~$S$ of~$\cT$}
{Let~\cs{$x^\dagger\notin X$}. In each~$T\in\cT$, if~$T'$ is the pendant subtree of~$T$ that~$S$ is a refinement of, replace~$T'$ by a single leaf labelled~$x^\dagger$. Remove the labels labelling leaves of~$S$ from~$X$ and add~$x^\dagger$ to~$X$.\\
\textbf{go to} Line 1}}
{\textbf{Chain Reduction:}
\For{$q=t-1,t-2,\ldots ,0$}
{
\If{there exists a maximal common $q$-star chain $(x_1,\ldots ,x_p)$ of~$\cT$ with~$p > (5k)^{t-q}$}
{Delete leaves $x_{(5k)^{t-q}+1},\ldots ,x_p$ from~$X$ and from each tree in~$\cT$ and repeatedly suppress outdegree-1 vertices and delete unlabelled outdegree-0 vertices until no such vertices remain.\\
\textbf{go to} Line 1}}
}
\caption{Kernelization algorithm for~$t:=|\cT|$ trees\label{alg:kernel}}
\end{algorithm}

Now assume that there exist leaves~$x,y$ that have a common parent in each tree in~$\cT$. We show how the common pendant subtree on~$x$ and~$y$ can be extended to a \emph{maximal} common pendant subtree of~$\cT$. Let~$\cT'$ be the result of modifying each tree in~$\cT$ by removing~$y$, suppressing the former parent of~$y$ if it gets outdegree~1 and relabelling~$x$ to~$z$ (with~$z\notin X$). Search, recursively, for a nontrivial maximal common pendant subtree of~$\cT'$. Let~$S'$ be such a subtree if it exists and, otherwise, let~$S'$ be the subtree consisting only of leaf~$z$. If~$S'$ does not contain~$z$, then let~$S:=S'$. If~$S'$ does contain~$z$, then let~$S$ be the result of adding two leaves~$x$ and~$y$ with edges~$(z,x)$ and~$(z,y)$ and removing the label of~$z$. Then,~$S$ is a nontrivial maximal common pendant subtree of~$\cT$.

\revB{Checking for each pair of taxa whether they have a common parent in each tree takes $O(|X|^2|\cT|)$ time. This has to be repeated at most $|X|$ times because one leaf is deleted in each iteration. Hence, the total running time is $O(|X|^3|\cT|)$.}
\end{proof}

\revB{We now turn to the chain reduction, and start by formally defining a chain.}
\medskip
\begin{definition}\label{def:chain}
If~$T$ is a tree on~$X$, $p\geq 2$ and~$x_1,\ldots ,x_p\in X$, then~$(x_1,\ldots ,x_p)$ is a \emph{chain} of~$T$ if:

\begin{enumerate}
\item[(1)] there exists a directed path $(v_1,...,v_\tau)$ in T, for some $\leo{\tau}\geq 1$;
\item[(2)] each $x_i$ is a child of some~$v_j$;
\item[(3)] if $x_i$ is a child of~$v_j$ and $i < p$, then $x_{i+1}$ is either a child
of $v_j$ or of $v_{j+1}$;
\item[(4)] for~$i\in\{2,\ldots ,\tau-1\}$, the children of~$v_i$ are all in~$\{v_{i+1},x_1,x_2,\ldots ,x_p\}$.
\end{enumerate}

If, in addition,~$\tau=1$ or the children of~$v_\tau$ are all in~$\{x_1,\ldots ,x_p\}$, then~$(x_1,\ldots ,x_p)$ is said to be a \emph{pendant chain} of~$T$. The \emph{length} of the chain is~$p$.
\end{definition}

\revB{Observe that $(3)$ can equivalently be replaced by
\begin{enumerate}
\item[(3')] $x_i$ is above $x_j$ in~$T$ whenever $i<j$.
\end{enumerate}}

\begin{figure}[t]
    \centering
    \includegraphics[scale=.35]{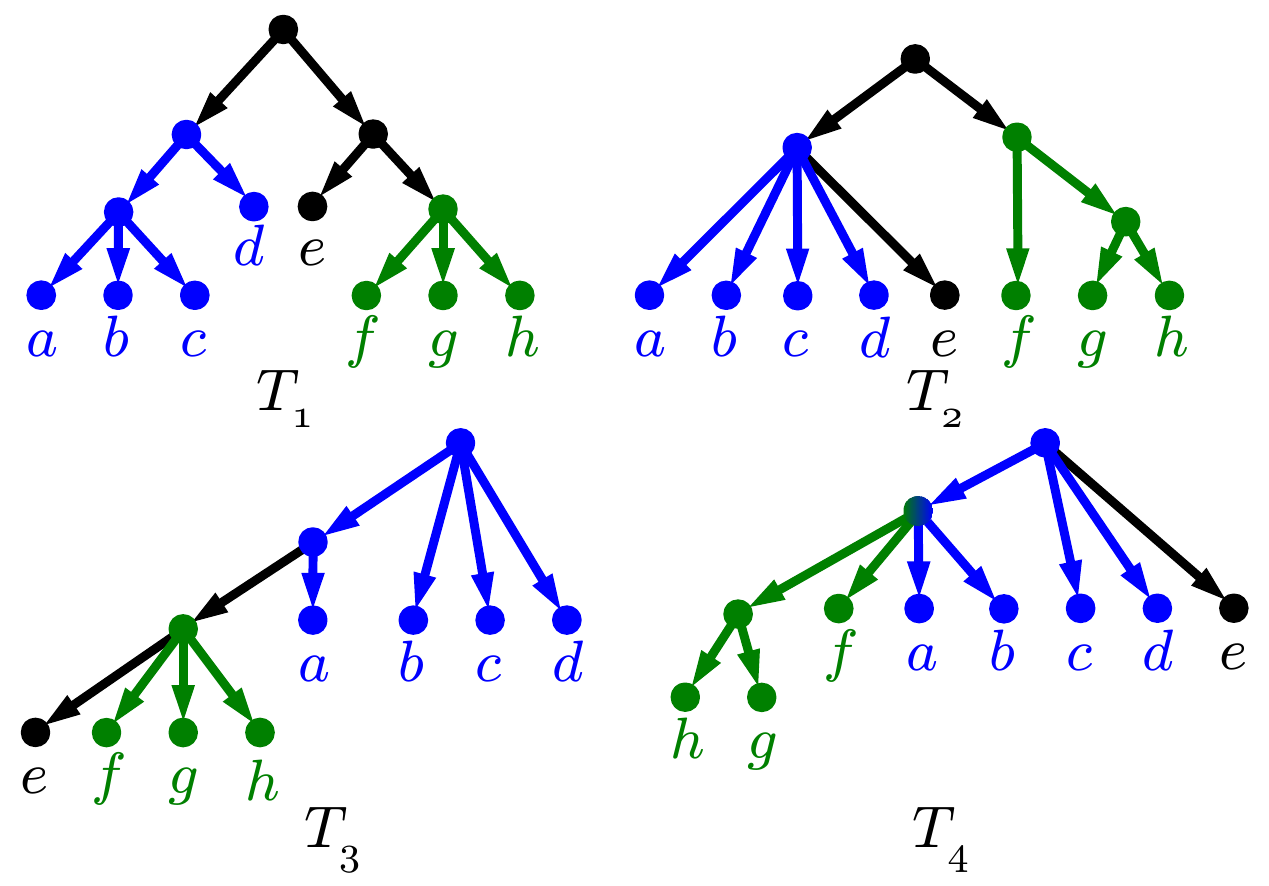}
    \caption{Example instance of \HN consisting of four trees that have a common pendant subtree on~$\{f,g,h\}$ and a common 1-star chain~$(d,c,b,a)$. Chain~$(d,c,b,a)$ is pendant in~$T_1$ and~$T_2$ \leo{(because $\tau=1$ in~$T_2$)} but not in~$T_3$ and~$T_4$. It is a 1-star chain because all its leaves have a common parent in only~$T_2$.\label{fig:chain}
    }
\end{figure}

A chain is said to be a \emph{common chain} of~$\cT$ if it is a chain of each tree in~$\cT$. The following observations follow easily from the definition of a chain.
\medskip
\begin{observation}\label{obs:subchain}
If~$(x_1,\ldots ,x_p)$ is a common chain of~$\cT$ and~$1\leq i < j \leq p$, then $(x_i,\ldots ,x_j)$ is a common chain of~$\cT$.
\end{observation}
\medskip
\begin{observation}\label{obs:commonparent}
If~$(x_1,\ldots ,x_p)$ is a chain of a tree~$T$,~$1\leq i < j \leq p$ and~$x_i$ and~$x_j$ have a common parent in~$T$, then~$x_i,\ldots ,x_j$ have a common parent in~$T$.
\end{observation}
\medskip

\revB{It will turn out that chains are easier to deal with when they form a star in more input trees, which leads to the following definition.}

\medskip
\begin{definition}\label{def:qstar}
If~$\cT$ is a set of trees on~$X$ and~$x_1,\ldots ,x_p\in X$, then~$(x_1,\ldots ,x_p)$ is a \emph{common $q$-star chain} of~$\cT$ if:
\begin{enumerate}
\item[(a)] $(x_1,\ldots ,x_p)$ is a \leo{common chain of~$\cT$} and
\item[(b)] in precisely~$q$ trees of~$\cT$, all of~$x_1,\ldots ,x_p$ have a common parent.
\end{enumerate}
\end{definition}

We say that a common $q$-star chain $(x_1,\ldots ,x_p)$ of~$\cT$ is \emph{maximal} if there is no common $q$-star chain $(y_1,\ldots ,y_{p'})$ of~$\cT$ with $\{x_1,\ldots ,x_p\} \subsetneq \{y_1,\ldots ,y_{p'}\}$. Notice that a common 0-star chain is a common chain that does not form a star in any tree. An illustration of the above definitions is in Figure~\ref{fig:chain}.

To prove correctness of the chain reduction, we use two central lemmas (Lemmas~\ref{lem:starchain} and~\ref{lem:spanned} below). The idea of these lemmas is illustrated in Figure~\ref{fig:central}. The two trees~$T_1$ and~$T_2$ in this figure have a common chain~$(a,b,c,d,e)$. Both trees are displayed by network~$N$. However, the leaves of the chain are spread out over different sides of the underlying generator~$G$ of~$N$. \cs{The idea of the correctness proof for} the chain reduction \cs{is} to argue that there exists a modified network~$N'$ in which the leaves of the chain~$(a,b,c,d,e)$ all lie on the same side. Moreover, network~$N'$ should display all input trees and its reticulation number should not be higher than the reticulation number of~$N$.

In~$T_1$, all leaves of the chain have a common parent \cs{and hence the chain is pendant in~$T_1$}. For this case, Lemma~\ref{lem:starchain} argues that all leaves of the chain can be moved to any \cs{edge} side that contains at least one of its leaves, and the resulting network still displays~$T_1$.

In~$T_2$, there are two leaves~$x_i=d$ and~$x_j=e$ that are on the same side of~$G$ (the \revA{blue} side~$s_b$) and that do not have a common parent in~$T_2$. For this case, Lemma~\ref{lem:spanned} argues that all the leaves of the chain can be moved to side~$s_b$, and the resulting network will still display~$T_2$. (Note that we cannot move all the leaves of the chain to \cs{side~$s_r$}, even though it contains two leaves~$b,c$ of the chain, because~$b$ and~$c$ have a common parent in~$T_2$.)

Hence, the network~$N'$ obtained by moving all leaves of the chain to the blue side~$s_b$ displays both~$T_1$ and~$T_2$. Furthermore, $r(N')=r(N)=2$.

We now formally state and prove the two aforementioned central lemmas. \revB{These lemmas use the following definition. \emph{Regrafting} a chain $(x_1,\ldots ,x_p)$ \emph{above} an indegree-1 vertex $\hat{v}$ of a network~$N$ means deleting~$x_1,\ldots ,x_p$ from~$N$, subdividing the edge entering~$\hat{v}$ by a directed path~$v_1,\ldots ,v_{p}$ and adding the leaves~$x_1,\ldots ,x_p$ by edges~$(v_1,x_1),$ $\ldots ,(v_p,x_p)$.}

\begin{figure}[t]
    \centering
    \includegraphics[scale=.35]{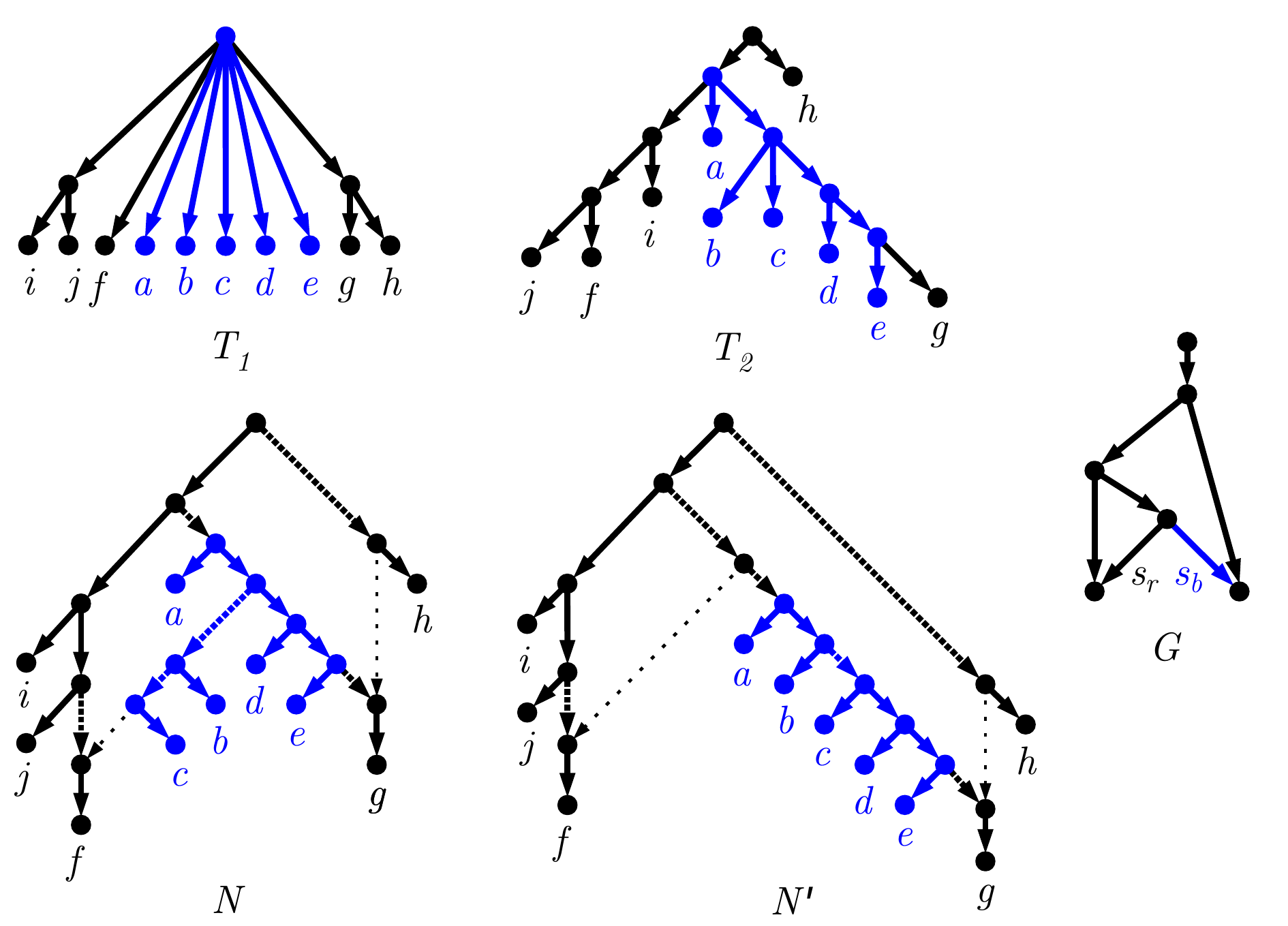}
    \caption{\cs{Two trees~$T_1$ and~$T_2$ with a common chain~$(a,b,c,d,e)$ highlighted in \revA{blue}, a network~$N$ that displays these trees, the network~$N'$ as constructed in Lemmas~\ref{lem:starchain} and~\ref{lem:spanned}, and the underlying generator~$G$ of both networks. Dashed and dotted edges are used to indicate that~$T_2$ can be obtained from either of~$N$ and~$N'$ by deleting the dotted edges and contracting the dashed edges.}\label{fig:central}
    }
\end{figure}
\medskip
\begin{lemma}\label{lem:starchain}
Let~$N$ be a binary \revA{network and} let~$T$ be a tree displayed by~$N$. Suppose that~$(x_1,\ldots ,x_p)$ is a \cs{pendant} chain of~$T$ and that~\revA{$\hat{v}$ is the parent of~$x_i$ in~$N$, with}~$i\in\{1,\ldots ,p\}$ and that~\revA{$\hat{v}$} has indegree~1. Let~$N'$ be the network obtained \revB{from~$N$ by regrafting $(x_1,\ldots ,x_p)$ above~$\hat{v}$ and cleaning up.} Then,~$N'$ displays~$T$.
\end{lemma}
\begin{proof}
\revB{Let~$D$ be the directed graph obtained from~$N$ by regrafting $(x_1,\ldots ,x_p)$ above~$\hat{v}$ (without cleaning up). Let~$v_i$ be the parent of~$x_i$ in~$D$ for $1\leq i\leq p$.} We show that~$D$ displays~$T$, from which it will follow directly that~$N'$ displays~$T$.

Since~$N$ displays~$T$, there is a subgraph~$T_s$ of~$N$ that is a refinement of~$T$. Clearly,~$T_s$ contains~$\revA{\hat{v}}$. Let~$T'_s$ be the subgraph of~$D$ obtained from~$T_s$ by \revB{regrafting $(x_1,\ldots ,x_p)$ above~$\hat{v}$} and repeatedly removing all unlabelled vertices that have outdegree~0 in~$T_s$.

It remains to prove that the resulting subgraph~$T_s'$ of~$D$ is a refinement of~$T$. \revB{First note that all leaves that are reachable from~$\hat{v}$ in~$T_s$ are in $\{x_1,\ldots ,x_p\}$ because this chain is pendant and~$N$ \revBnew{is} binary. Therefore, the chain is also pendant in~$T_s'$. We can find a set of edges of~$T_s'$ such that contracting these edges gives~$T$ as follows.} Let~$v^*$ be the lowest common ancestor of~$x_1,\ldots ,x_p$ in~$T_s$. Then there is a directed path from~$v^*$ to~$\revA{\hat{v}}$ in~$T_s$ (possibly, $v^*=\revA{\hat{v}}$). Since~$T_s$ is a refinement of~$T$, there is a set of edges~$E$ such that~$T$ can be obtained from~$T_s$ by contracting all edges of~$E$. Assume without loss of generality that~$E$ does not contain any edges whose head is a leaf. Let~$E'$ be the set of edges of~$D$ containing all edges of~$E$ that are in~$T_s'$ and all edges on the directed path from~$\revB{v^*}$ to~$v_1$ in~$T_s'$. Then,~$T$ can be obtained from~$T_s'$ by contracting the edges of~$E'$ and part of the edges on the directed path~$v_1,\ldots , v_p,
\revA{\hat{v}}$ \revB{(a part of these edges because the chain is not necessarily a star in~$T$)}. Hence,~$D$ displays~$T$.\end{proof} 

\revA{Note that, if network~$N$ in Lemma~\ref{lem:starchain} has no nontrivial pendant subtrees, the assumption in the statement of the lemma that~$\revA{\hat{v}}$ has indegree~1 is equivalent to assuming that~$x_i$ is on an edge side of the underlying generator.}
\medskip
\begin{lemma}\label{lem:spanned}
Let~$N$ be a binary network without nontrivial pendant subtrees and let~$T$ be a tree displayed by~$N$. Suppose that~$(x_1,\ldots ,x_p)$ is a chain of~$T$, \cs{that~$1\leq i,j\leq p$,} that~$x_i$ and~$x_j$ do not have a common parent in~$T$ and that~$x_i$ and~$x_j$ are on the same side of~$N$. Let~$\revA{\hat{v}}$ be the parent of~$x_i$ in~$N$. \revB{Let~$N'$ be the network obtained from~$N$ by regrafting $(x_1,\ldots ,x_p)$ above}~$\revA{\hat{v}}$ and cleaning up. Then,~$N'$ displays~$T$.
\end{lemma}
\begin{proof}
\cs{First note that, since~$x_i$ and~$x_j$ are on the same side of~$N$, this must be an edge side. Hence, the indegree of~$\revA{\hat{v}}$ is~1 (as in Lemma~\ref{lem:starchain}).}

\revB{Let~$D$ be the directed graph obtained from~$N$ by regrafting $(x_1,\ldots ,x_p)$ above~$\hat{v}$ (without cleaning up). Let~$v_i$ be the parent of~$x_i$ in~$D$ for $1\leq i\leq p$. We show (as in the proof of Lemma~\ref{lem:starchain}) that~$D$ displays~$T$, from which it will follow directly that~$N'$ displays~$T$.}

Since~$N$ displays~$T$, there is a subgraph~$T_s$ of~$N$ that is a refinement of~$T$. Clearly,~$T_s$ contains~$\revA{\hat{v}}$. Let~$T'_s$ be the subgraph of~$D$ obtained from~$T_s$ \revB{by regrafting~$(x_1,\ldots ,x_p)$ above}~$\revA{\hat{v}}$ and repeatedly removing all unlabelled vertices that have outdegree~0 in~$T_s$.

It remains to prove that the resulting subgraph~$T_s'$ of~$D$ is a refinement of~$T$. \revC{As in the proof of Lemma~\ref{lem:starchain}}, since~$T_s$ is a refinement of~$T$, there is a set of edges~$E$ such that~$T$ can be obtained from~$T_s$ by contracting all edges of~$E$. Assume without loss of generality that~$E$ does not contain any edges whose head is a leaf. \cs{Let~$v^*$ be the lowest common ancestor of~$x_1,\ldots ,x_p$ in~$T_s$. Let~$E'$ be the set of edges of~$D$ containing all edges of~$E$ that are in~$T_s'$ and all edges on the directed path from~$v^*$ to~$v_1$ in~$D$.} We distinguish two cases.

First suppose that~$(x_1,\ldots ,x_p)$ is a pendant chain \cs{(of~$T$)}. Then,~$D$ displays~$T$ \cs{by Lemma~\ref{lem:starchain}}.

Now suppose that~$(x_1,\ldots ,x_p)$ is not a pendant chain \cs{(of~$T$)}. Let~$v^{**}$ be the parent of~$x_p$ in~$T_s$ and let~$P^*$ be the directed path from~$v^*$ to~$v^{**}$ in~$T_s$. \revA{Since $x_i$ and~$x_j$ do not have a common parent in~$T$, they must also have distinct parents in~$T_s$. Moreover, since~$x_i$ and~$x_j$ are on the same side of~$N$, their parents must lie on~$P^*$.} Since the chain is not pendant,~$v^{**}$ has at least one child that has at least one leaf-descendant in~$T_s$ that is not in~$\{x_1,\ldots ,x_p\}$. Therefore, \revA{a subdivision of} the path~$P^*$ is preserved in~$T_s'$. Hence, as before,~$T$ can be obtained from~$T_s'$ by contracting the edges of~$E'$ and part of the edges on the directed path~$v_1,\ldots ,v_p,\revA{\hat{v}}$. Hence,~$N'$ displays~$T$.
\end{proof}

The next lemma shows correctness of the chain reduction, and thereby of Algorithm~\ref{alg:kernel}. It is based on the idea that, if a \revB{$q$-star chain is long enough and $q'$-star chains for~$q'>q$ have already been reduced}, then one of Lemmas~\ref{lem:starchain} and~\ref{lem:spanned} applies for each tree. \revB{Lemma~\ref{lem:starchain} applies to each tree in which the $q$-star chain is a star. In the other trees, not too many leaves can have a common parent because $q'$-star chains for~$q'>q$ have already been reduced, which will make it possible to apply Lemma~\ref{lem:spanned}.}
\medskip
\begin{lemma}\label{lem:kernel}
Let~$q\in\{0,\ldots ,t-1\}$ (with~$t=|\cT|$) and let~$(X,\cT,k)$ be an instance of \HN without nontrivial common pendant subtrees or maximal common $q'$-star chains of more than $(5k)^{t-q'}$ leaves, for~$q<q'\leq t-1$. Let $(X',\cT',k)$ be the instance obtained after applying the chain reduction to a maximal common $q$-star chain $C=(x_1,\ldots ,x_p)$ of~$\cT$ with~$p>(5k)^{t-q}$. Then~$r(\cT)\leq k$ if and only if~$r(\cT')\leq k$.
\end{lemma}
\begin{proof}
It is clear that if~$r(\cT)\leq k$ then~${r(\cT')\leq k}$ because the chain reduction only deletes leaves (and suppresses and deletes vertices).

It remains to prove the other direction. Assume that~$r(\cT')\leq k$, \revB{that is}, there exists a network~$N'$ that displays~$\cT'$ and has~$r(N')\leq k$. Define~$m:=(5k)^{t-q-1}$. Hence, there are no common chains of~$\cT$ of more than~$m$ leaves that have a common parent in more than~$q$ of the trees.

Let~$C'=(x_1,\ldots ,x_{5km})$. First observe that~$C'$ is a common chain of~$\cT'$ and, moreover, that~$C'$ is a common $q'$-star chain of~$\cT'$ with~$q'\geq q$. Moreover, we claim the following.
\medskip
\begin{claim}\label{claim:commonparent}
Any two leaves in~$\{ x_1,\ldots ,x_{5km-1}\}$ have a common parent in a tree~$T\in\cT$ if and only if they have a common parent in the corresponding tree~$T'\in\cT'$.
\end{claim}

This claim follows directly from the observation that, in the chain reduction, the parents of the leaves~$x_1,\ldots ,x_{5km-1}$ cannot become outdegree-1 and are therefore not being suppressed. Correctness of the next claim can be verified in a similar way.

\medskip
\begin{claim}\label{claim:pendant}
If~$C$ is not pendant in~$T\in\cT$, then any two leaves in~$\{ x_1,\ldots ,x_{5km}\}$ have a common parent in~$T$ if and only if they have a common parent in the corresponding tree~$T'\in\cT'$.
\end{claim}

\revB{The above claim is not true for pendant chains when all but one of the children of the parent of~$x_{5km}$ are deleted by the chain reduction. However, this is not a problem because we only need the claim for non-pendant chains.}

Now define $$C^*:=(x_1,x_{1+m},x_{1+2m},\ldots ,x_{1+(5k-1)m}),$$
\revB{that is},~$C^*$ contains~$5k$ leaves and the indices of any two subsequent leaves are~$m$ apart.

Let~$G'$ be the generator underlying~$N'$. Each leaf of~$C^*$ is on a certain side of~$G'$. Since~$G'$ has at most~$5k-1$ sides (by Lemma~\ref{lem:sides}) and~$C^*$ contains~$5k$ leaves, there exist two leaves~$x_i,x_j$ of~$C^*$ \cs{(and thus of~$C'$)} that are on the same \cs{edge} side of~$G'$ by the pigeonhole principle. Assume without loss of generality that~$j>i$. Then, by the construction of~$C^*$, $j\geq i + m$. 

We modify network~$N'$ to a network~$N''$ by \revB{regrafting} chain~$C'$ \revB{above the parent~$\hat{v}$ of $x_i$ in~$N'$ and cleaning up.}

For each tree~$T'\in\cT'$ in which all of $x_1,\ldots ,x_{5km}$ have \leo{a} common parent, Lemma~\ref{lem:starchain} shows that~$N''$ displays~$T'$. There are at least~$q$ such trees. In fact, the following claim \cs{implies} that there are precisely~$q$ such trees. Moreover, the claim shows that in all other trees~$x_i$ and~$x_j$ do \emph{not} have a common parent. Therefore, it follows from Lemma~\ref{lem:spanned} that these trees are also displayed by~$N''$.
\medskip
\begin{claim}\label{claim:atmostq}
The number of trees of~$\cT'$ in which~$x_i$ and~$x_j$ have a common parent is at most~$q$.
\end{claim}

To prove the claim, consider~$C^{**}:=(x_i,\ldots ,x_j)$. Since~$C^{**}$ is a subchain of~$C'$, it is a chain of each tree in~$\cT'$ by Observation~\ref{obs:subchain}.

First consider the case~$q=t-1$ and assume that~$x_i$ and~$x_j$ have a common parent in more than~$q$ trees in~$\cT'$ and hence in all trees in~$\cT'$. Then,~$x_i,\ldots ,x_j$ all have a common parent in all trees in~$\cT'$, by Observation~\ref{obs:commonparent}. Since~$C$ is a $q$-star chain of~$\cT$, there are~$q=t-1$ trees in~$\cT$ in which all leaves of~$C$ have a common parent. Let~$T^*$ be the only tree in~$\cT$ in which the leaves of~$C$ do not all have a common parent. Then~$C$ is not pendant in~$T^*$ or its leaves would form a nontrivial common pendant subtree of~$\cT$. Hence,~$x_i$ and~$x_j$ have a common parent in~$T^*$ by Claim~\ref{claim:pendant}. However, \revC{then~$x_i,~x_j$, and their common parent form a nontrivial common pendant subtree of~$\cT$. This is a contradiction to the assumption that~$\cT$ has no nontrivial common pendant subtrees.}

To finish the proof of Claim~\ref{claim:atmostq}, consider the case $q<t-1$. \revA{In this case, $m>1$. Recall that $x_j$ is in~$C^*$ and that the highest-indexed element of~$C^*$ is~$x_{1+(5k-1)m}$. Hence, $m>1$ implies that $5km > 1 + (5k-1)m$ and hence that}~$C^{**}$ contains only leaves in~$\{x_1,\ldots ,x_{5km-1}\}$. Because~$C^{**}$ contains more than~$m$ leaves, the number of trees of~$\cT$ in which all the leaves of~$C^{**}$ have a common parent is at most~$q$ \leo{(here we use the fact that there are no common $q'$-star chains for~$q'>q$ that have more than~$m$ leaves)}. Hence, it follows from Claim~\ref{claim:commonparent} that the number of trees of~$\cT'$ in which the leaves of~$C^{**}$ have a common parent is at most~$q$. Claim~\ref{claim:atmostq} then follows by Observation~\ref{obs:commonparent}.

Hence, we have shown that~$N''$ displays~$\cT'$. We now construct a network~$N$ from~$N''$ by replacing the reduced chain by the unreduced chain. More precisely, let~$e_{5km}$ be the edge of~$N''$ that leaves~$v_{5km}$ but is not the edge~$(v_{5km},x_{5km})$. Subdivide~$e_{5km}$ by a directed path $(v_{5km+1},\ldots ,v_p)$ and add leaves~$x_{5km+1},\ldots ,x_p$ by edges $(v_{5km+1},x_{5km+1}),$ $\ldots ,$  $(v_p,x_p)$. This gives~$N$. Then, by a similar argument as in the proof of Lemma~\ref{lem:spanned},~$N$ displays~$\cT$. Moreover, since none of the applied operations increase the reticulation number, we have~$r(N)\leq r(N')$.
\end{proof}

The next lemma shows that the chain reduction can be performed in polynomial time. It uses the following additional definitions \revA{and observation}. \revB{We define an $s$-$t$-\emph{chain} of a tree~$T$ as a chain $(x_1,\ldots ,x_p)$ of~$T$ with~$x_1=s$ and~$x_p=t$. A set~$C\subseteq X$ of leaves is $s$-$t$-\emph{chainable} (\revB{with respect to} a set~$\cT$ of trees on~$X$) if there exists an ordering of the leaves in~$C$ that is a common $s$-$t$-chain for~$\cT$. We first show that this ordering is unique when~$\cT$ has no nontrivial common pendant subtrees.}

\medskip

\begin{observation}\label{obs:unique} \revA{If~$\cT$ has no nontrivial common pendant subtrees and~$C$ is $s$-$t$-chainable \revB{with respect to}~$\cT$, then there exists a unique common $s$-$t$-chain of~$\cT$ with leaf set~$C$.}
\end{observation}
\begin{proof}\revA{Assume to the contrary that there are two distinct common $s$-$t$-chains of~$\cT$ with leaf set~$C$. Then there exist two leaves~$x,y\in C$ such that~$x$ comes before~$y$ in one of these chains and after~$y$ in the other one. Then, in each~$T\in \cT$, $x$ is above~$y$ and~$y$ is above~$x$, implying that~$x$ and~$y$ have a common parent. This contradicts the assumption that~$\cT$ has no nontrivial common pendant subtrees.}
\end{proof}
\medskip
\begin{lemma}\label{lem:polychain}
There exists \revB{an $O(|X|^6|\cT|)$-time} algorithm that, given a set~$\cT$ of trees on~$X$ and~$q\in\NN$, decides if there exists a common $q$-star chain of~$\cT$ and constructs such a chain of maximum size if one exists.
\end{lemma}
\begin{proof}
First, we show how to decide in \revB{polynomial time} if a set~$C\subseteq X$ is \revB{$s$-$t$-}chainable. It is easy to check whether \revB{parts}~(1), (2) and~(4) of Definition~\ref{def:chain} are satisfied in each tree in~$\cT$ \revB{and whether~$x_1=s$ and~$x_p=t$. Hence, assume that this is the case. \revB{Part}~(3) \revB{is then equivalent to (3'), which} requires that there exists an ordering $(x_1,\ldots ,x_p)$ of the leaves of~$C$ such that, whenever $i < j$, $x_i$ is above~$x_j$ in each tree~$T\in\cT$. To check if there exists such an ordering, construct a directed graph~$D_C$ with vertex set~$C$ and an edge~$(x,y)$ precisely if there exists a tree~$T\in\cT$ in which~$x$ is strictly above~$y$. An ordering $(x_1,\ldots ,x_p)$ of the leaves in~$C$ is called a \emph{topological ordering} for~$D_C$ if there is no edge $(x_j,x_i)$ with $i<j$. Hence, $(x_1,\ldots ,x_p)$ is a topological ordering of~$D_C$ if and only if, whenever $i<j$, there is no tree~$T\in\cT$ in which~$x_j$ is strictly above~$x_i$, \revB{that is}, $x_i$ is above~$x_j$ in all trees~$T\in\cT$. Since it is well known that a directed graph is acyclic if and only if there exists a topological ordering of its vertices, it follows that $D_C$ is a acyclic if and only if~$C$ is $s$-$t$-chainable.}

We \cs{now} describe an algorithm that decides, for~$s,t\in X$, if there exists a common $q$-star $s$-$t$-chain of~$\cT$ and constructs such a chain of maximum size if one exists. This proves the lemma because one can simply try each combination $s,t\in X$.

The algorithm is as follows. If, in at least one~$T\in\cT$, $s$ is not above $t$ then there exists no $s$-$t$-chain of~$T$ and we stop. Otherwise, let~$P_T$ be the directed path from~$p_T(s)$ to~$p_T(t)$ in~$T$. If $|\{T\in\cT \mid p_T(s) = p_T(t) \}|\neq q$, then no $s$-$t$-chain is a $q$-star chain and we stop. Otherwise, define~$C$ as the set of leaves containing~$s$,~$t$ and all leaves~$x\in X$ for which~$p_T(x)$ is an internal vertex of~$P_T$ for at least one~$T\in\cT$. Note that any common $s$-$t$-chain must contain all leaves in~$C$. Moreover, if~$C$ is not \revB{$s$-$t$-}chainable then \revB{either it contains a leaf that has its parent not on the path~$P_T$ in some~$T\in\cT$, or the directed graph~$D_C$ is cyclic. In either case,} no superset of~$C$ is \revB{$s$-$t$-}chainable. Hence, if~$C$ is not \revB{$s$-$t$-}chainable, then~$\cT$ has no common $s$-$t$-chain and we stop. Otherwise, there exists a unique common chain on~$C$. In that case, we try to add as many leaves as possible to~$C$ such that it remains \revB{$s$-$t$-}chainable (because we need a chain of maximum size). Define
\[
X' := \{x\in X\setminus C \hspace{5pt}\mid\hspace{5pt} p_T(x)\in\{p_T(s),p_T(t)\} \hspace{5pt} \forall \hspace{2pt} T\in\cT \mbox{ and } C\cup\{x\} \mbox{ is chainable}\}.
\]
The set~$X'$ contains all leaves of~$X\setminus C$ that can be in a common $s$-$t$-chain of~$\cT$. Now consider the directed graph~$D=(X',A)$ with an edge~$(x,y)\in A$ precisely \revA{if~$x$ is above~$y$} in all~$T\in\cT$. \revB{Observe that~$D$ is acyclic because, if it had a directed cycle, the vertices in the cycle would correspond to leaves with a common parent in all trees~$T\in\cT$, contradicting the assumption that~$\cT$ has no nontrivial common pendant subtrees.}

We claim that a set~$X''\subseteq X'$ forms a directed path in~$D$ if and only if~$C\cup X''$ is chainable. First suppose that~$C\cup X''$ is chainable. Then let~$(x_1,\ldots ,x_n)$ be an ordering of the elements of~$X''$ following the ordering in the chain (from top to bottom). Then there is, in each~$T\in\cT$, a directed path from~$p_T(x_i)$ to~$p_T(x_{i+1})$ ($i=1,\ldots ,n-1$) and hence $(x_1,\ldots ,x_n)$ is a directed path in~$D$. To show the converse, assume that $(x_1,\ldots ,x_n)$ is a directed path in~$D$. By the definition of~$X'$, for each~$x\in X'$ \revA{it} holds that~$C\cup\{x\}$ is chainable. Moreover, the position where~$x$ can be inserted into \cs{the chain on}~$C$ is unique \revB{by Observation~\ref{obs:unique}}. Hence, the leaves~$x_1,\ldots ,x_n$ can be inserted into \cs{the chain on}~$C$ one by one, where their relative position is determined by the order $x_1,\ldots ,x_n$. This gives a common chain of~$\cT$ and hence~$C\cup X''$ is chainable.

A longest directed path in~$D$ can be found in polynomial time since~$D$ is acyclic. Let~$X_P$ be the set of vertices on such a longest path. Then $C\cup X_P$ is chainable and the corresponding common $s$-$t$-chain of~$\cT$ is a common $q$-star $s$-$t$-chain of maximum size. 

\revB{It remains to analyse the running time. Checking if a set is chainable takes $O(|X|^3|\cT|)$ time. Hence, constructing the set~$X'$ takes~$O(|X|^4|\cT|)$ time. Constructing the graph~$D$ takes $O(|X|^3|\cT|)$ time and searching for a longest path in~$D$ takes $O(|X|^2)$ time. Hence, the running time is dominated by the construction of the set~$X'$. Since this is done for each $s,t\in X$, the total running time is $O(|X|^6|\cT|)$.}\end{proof}

To see that Algorithm~\ref{alg:kernel} runs in polynomial time, it remains to observe that at least one leaf is removed in each iteration and hence that the number of iterations is bounded by~$|X|$. \revB{Since the running time is dominated by the search for maximum common $q$-star chains, it follows from Lemma~\ref{lem:polychain} that Algorithm~\ref{alg:kernel} runs in $O(|X|^7|\cT|)$ time.}

It remains to bound the size of the kernel.
\medskip
\begin{lemma}\label{lem:kernelsize}
Let $(X,\cT,k)$ be an instance of \HN with~$k\geq 1$ and $r(\cT)\leq k$. If~$(X',\cT',k)$ is the instance obtained after applying Algorithm~\ref{alg:kernel}, then $|X'|\leq 4k(5k)^{|\cT|}$.
\end{lemma}
\begin{proof}
Since~$r(\cT)\leq k$, also $r(\cT')\leq k$ by Lemma~\ref{lem:kernel}. Hence, there exists a \revA{network}~$N'$ on~$X'$ that displays~$\cT'$ and has~$r(N')\leq k$. We can and will assume that~$N'$ is binary (by Observation~\ref{obs:binary}). Since subtree reductions have been applied in Algorithm~\ref{alg:kernel},~$\cT'$ does not contain any nontrivial common pendant subtrees. Hence,~$N'$ does not contain any nontrivial pendant subtrees. Let~$G'$ be the generator underlying~$N'$. By Lemma~\ref{lem:sides},~$G'$ has at most~$k$ vertex sides and at most~$4k-1$ edge sides. Each vertex side contains exactly one leaf. Since chain reductions have been applied in Algorithm~\ref{alg:kernel},~$\cT'$ does not contain any common chains of more than~$(5k)^{|\cT|}$ leaves. Hence, each edge side of~$G'$ contains at most $(5k)^{|\cT|}$ leaves. Therefore,
$$|X'| \leq k + (4k-1)(5k)^{|\cT|} \leq 4k(5k)^{|\cT|}.$$
\end{proof}

Correctness of the following theorem now follows from Lemmas~\ref{lem:subtreereduction}--\ref{lem:polysubtree} and~\ref{lem:kernel}--\ref{lem:kernelsize}.
\medskip
\begin{theorem}\label{thm:kernel}
The problem \HN on~$|\cT|=t$ trees admits a kernel with at most $4k(5k)^t$ leaves.
\end{theorem}

\section{A polynomial kernel for bounded outdegrees}\label{sec:boundedoutdegrees}

We now show that an approach similar to the one in the previous section can also be used to obtain a kernelization for \HN in the case that not the number of input trees but their maximum outdegree is bounded. Algorithm~\ref{alg:kernel2} describes a polynomial kernel for \HN if the maximum outdegree of the input trees is \revB{at most}~$\Delta^+$.

\begin{algorithm}[h]
\leo{Apply the subtree reduction (see Algorithm~\ref{alg:kernel}).}\\
{\textbf{Chain Reduction:}
\If{there is a maximal common chain $(x_1,\ldots ,x_p)$ of~$\cT$ with~$p > 5k(\Delta^+-1)$}
{Delete leaves $x_{5k(\Delta^+-1)+1},\ldots ,x_p$ from~$X$ and from each tree in~$\cT$ and repeatedly suppress outdegree-1 vertices and delete unlabelled outdegree-0 vertices until no such vertices remain.\\
\textbf{go to} Line 1}
}
\caption{Kernelization algorithm for bounded outdegree\label{alg:kernel2}}
\end{algorithm}

\revB{Algorithm~\ref{alg:kernel2} has the same running time as Algorithm~\ref{alg:kernel}, namely $O(|X|^7|\cT|)$.}

\medskip
\begin{theorem}\label{thm:kerneldegree}
The problem \HN on trees with maximum outdegree~$\Delta^+$ admits a kernel with at most $20k^2(\Delta^+-1)$ leaves.
\end{theorem}
\begin{proof}
We claim that Algorithm~\ref{alg:kernel2} provides the required kernelization. The algorithm can be applied in polynomial time by Lemmas~\ref{lem:polysubtree} and~\ref{lem:polychain} (see above). Correctness of the subtree reduction has been shown in Lemma~\ref{lem:subtreereduction} and the proof of the kernel size is analogous to the proof of Lemma~\ref{lem:kernelsize}. Hence, it remains to show correctness of the chain reduction.

Let~$(X,\cT,k)$ be an instance of \HN, let $C=(x_1,\ldots ,x_p)$ be a maximal common chain of~$\cT$ and let $(X',\cT',k)$ be the instance obtained by reducing this chain to~$C'=(x_1,\ldots ,x_{5k(\Delta^+-1)})$. As in the proof of Lemma~\ref{lem:kernel}, the only nontrivial direction of the proof is to show that if $r(\cT')\leq k$ then $r(\cT)\leq k$.

Assume that $r(\cT')\leq k$, \revB{that is}, that there exists a \revA{network}~$N'$ that displays~$\cT'$ and has~$r(N')\leq k$. Define $$C^*:=\{x_1 ,x_{1+(\Delta^+-1)}, x_{1+2(\Delta^+-1)}, \ldots ,x_{1+(5k-1)(\Delta^+-1)}\}.$$ Then, since~$C^*$ contains~$5k$ leaves and~$N'$ has at most~$5k-1$ sides, there are two leaves~$x_i,x_j$ in~$C^*$ that are on the same \cs{(edge)} side of~$N'$. Assume~$j>i$. By the construction of~$C^*$, $|\{x_i,\ldots ,x_j\}| > \Delta^+-1$. Consider a tree~$T'\in\cT'$ in which~$x_i$ and~$x_j$ have a common parent~$v$. Since~$v$ has outdegree at most~$\Delta^+$, its children are precisely~$x_i,\ldots ,x_j$. This means that~$C'$ is a pendant chain in~$T'$ (because, by the definition of pendant chain, any \revA{parent of a leaf} of a non-pendant chain has at least one child that is not a leaf or a leaf not in the chain). Hence, in each tree~$T'\in\cT'$, either~$x_i$ and~$x_j$ do not have a common parent or~$C'$ is a pendant chain.

\revB{Let $m=\Delta^+-1$.} We modify network~$N'$ to a network~$N''$ \revB{by regrafting $C'=(x_1,\ldots ,x_{5km})$ above the parent of~$x_i$ and cleaning up (see the definition of regrafting above Lemma~\ref{lem:starchain}). Let~$v_i$ be the parent of~$x_i$ in~$N''$, for $1\leq i \leq 5km$.}

For each tree~$T'\in\cT'$ in which~$C'$ is a pendant chain, it follows from \cs{Lemma~\ref{lem:starchain}} that~$N''$ displays~$T'$.

For each tree~$T'\in\cT'$ in which~$x_i$ and~$x_j$ do not have a common parent, it follows from Lemma~\ref{lem:spanned} that~$N''$ displays~$T'$ (using that~$x_i$ and~$x_j$ are on the same side in~$N'$).

Therefore,~$N''$ displays~$\cT'$. We now construct a network~$N$ from~$N''$ as follows. Let~$e_{5km}$ be the edge that leaves~$v_{5km}$ but is not the edge~$(v_{5km},x_{5km})$. Subdivide~$e_{5km}$ by a directed path $(v_{5km+1},\ldots ,v_p)$ and add leaves~$x_{5km+1},\ldots ,x_p$ by edges $(v_{5km+1},x_{5km+1}), \ldots ,$ $(v_p,x_p)$. Then,~$N$ displays~$\cT$. Moreover, since none of the applied operations increase the reticulation number, we have~$r(N)\leq r(N')$ and hence that $r(\cT)\leq k$.
\end{proof}

\section{An exponential-time algorithm}\label{sec:alg}

\leo{From the existence of the kernelization in Theorem~\ref{thm:kernel} it follows directly that there exists an FPT algorithm for \HN parameterized by~$k$ and~$t=|\cT|$. Nevertheless, we find it useful to describe an exponential-time algorithm for \HN which combined with the kernelization then gives an explicit FPT algorithm. We do this in this section. Although the running time of the algorithm is not particularly fast, the algorithm is nontrivial and it is not clear if a significantly faster algorithm exists for general instances of \HN.} \revB{Moreover, this is the first proof that \HN is in the class~$\textsf{XP}$.}

\leo{As before, we may restrict to binary networks by Observation~\ref{obs:binary}. Moreover, by Lemmas~\ref{lem:polysubtree} and~\ref{lem:subtreereduction} we may assume that~$\cT$ has no nontrivial common pendant subtrees and hence we may restrict to networks with no nontrivial pendant subtrees.}

\leo{We need a few additional definitions. Let~$G$ be the underlying generator of network~$N$ and let~$s$ be a side of~$G$. We use~$X_N(s)$ to denote the set of leaves that are on side~$s$ in network~$N$.
\revX{The} \emph{top leaf} on side~$s$ in~$N$ \revX{is} the leaf~$x_s^+\in X_N(s)$ for which there is no leaf in~$X_N(s)\setminus\{x_s^+\}$ that is above~$x_s^+$. Similarly, the \emph{bottom leaf} on side~$s$ in~$N$ is the leaf~$x_s^-\in X_N(s)$ for which all leaves in~$X_N(s)$ are above~$x_s^-$ \revX{(note that each leaf is above itself)}. For two sides~$s,s'$ of~$G$, we say that~$s'$ is \emph{below}~$s$ if there is a directed path from (the head of)~$s$ to (the tail of)~$s'$ in~$G$.}

\leo{We define a \emph{partial network} \revB{with respect to}~$X$ as a binary network~$N_p$ on~$X_p \subseteq X$ such that, for each side~$s$ of the underlying generator of~$N_p$, there are at most two leaves on side~$s$. If~$N$ is a network on~$X$ with no nontrivial pendant subtrees, then we say that \revA{a} partial network~$N_p$ \revB{with respect to}~$X$ is \emph{consistent} with~$N$ (and that~$N$ is consistent with~$N_p$) if
\begin{enumerate}
\item[(i)] $N$ and~$N_p$ have the same underlying generator~$G$;
\item[(ii)] for each side~$s$ of~$G$ with~$|X_N(s)|\leq 2$ \revB{it} holds that $X_{N_p}(s) = X_N(s)$;
\item[(iii)] for each side~$s$ of~$G$ with~$|X_N(s)|\geq 2$ \revB{it} holds that $X_{N_p}(s) = \{x_s^+,x_s^-\}$ with $x_s^+$ and~$x_s^-$ respectively the top and bottom leaf on side~$s$ in~$N$.
\end{enumerate}}

\leo{We can bound the number of partial networks as follows.}
\medskip
\begin{lemma}\label{lem:partialnetworks}
\leo{The number of partial networks~$N_p$ \revB{with respect to}~$X$ is at most $$2^{9k\log{k}+O(k)}(n+1)^{9k}$$ with ${k=r(N_p)}$ and~$n=|X|$.}
\end{lemma}
\begin{proof}
\cs{First, we bound the number of binary $k$-reticulation generators. The number of binary \revA{networks} with~$|V|$ vertices is at most $2^{\frac{3}{2}|V|\log{|V|}+O(|V|)}$ \cite{counting}. A binary $k$-reticulation generator has at most~$3k$ vertices (see e.g.~\cite{thesis}) and can be turned into a binary \revA{network} as follows. For each reticulation~$r$, if it has outdegree 0, we add a leaf~$x$ with an edge~$(r,x)$. Moreover, if~$r$ has two incoming parallel edges $(u,r)$, we subdivide one of them by a vertex~$w$ and add a leaf~$y$ with an edge~$(w,y)$. Per reticulation we have added at most~3 vertices. Hence, the total number of vertices is at most~$6k$. Hence, the number of binary $k$-reticulation generators is at most $2^{9k\log{k}+O(k)}$.}

Each binary $k$-reticulation generator~$G$ has~$k$ vertex sides and at most~$4k-1$ edge sides by Lemma~\ref{lem:sides}. Each vertex side contains exactly one leaf, for which there are~$|X|$ possibilities. Each edge side contains at most two leaves. Hence, there are $|X|+1$ possibilities for the top leaf (including the possibility that there is no top leaf) and $|X|+1$ possibilities for the bottom leaf (again, including the possibility that there is no such leaf). Therefore, for each generator~$G$, the number of partial networks is at most $|X|^k (|X|+1)^{2(4k-1)}$, which is at most $(|X|+1)^{9k}$.
\end{proof}

\leo{The idea of our approach is to loop through all partial networks \revB{with respect to}~$X$ that have reticulation number at most~$k$. If there exists some network~$N$ displaying~$\cT$ with~$r(N)\leq k$, then in some iteration, we will have a partial network~$N_p$ that is consistent with~$N$. We will now show how to extend~$N_p$ to~$N$ in polynomial time. We use the following notation. Given two leaves~$x,y\in X$, we write~$x\rightarrow y$ if there exists a tree~$T\in\cT$ with a vertex~$v$ such that there is a directed path in~$T$ from~$v$ to~$y$ but not from~$v$ to~$x$. Similarly, given three leaves~$x,y,z\in X$, we write $x\rightarrow y,z$ if there exists a tree~$T\in\cT$ with a vertex~$v$ such that there are directed paths in~$T$ from~$v$ to~$y$ and from~$v$ to~$z$ but not from~$v$ to~$x$. Note that we do not explicitly indicate the dependency of the~$\rightarrow$ relationship on~$\cT$ to improve readability. Also note that~$x\rightarrow y$ and~$x\rightarrow z$ does not imply~$x\rightarrow y,z$. We prove the following.}
\medskip
\leo{\begin{lemma}\label{lem:arrows}
Let~$\cT$ be a set of trees on~$X$ with no nontrivial common pendant subtrees, let~$N$ be a binary network that displays~$\cT$, let~$G$ be its underlying generator, let~$s$ be a side of~$G$, let $x_s^+$ and~$x_s^-$ be, respectively, the top and bottom leaf on side~$s$ in~$N$ and let~$x,y\in X\setminus\{x_s^+,x_s^-\}$. Then,
\begin{enumerate}
\item[(a)] if $x$ is on side~$s$ then~$x_s^+ \rightarrow x, x_s^-$;
\item[(b)] if~$x_s^+ \rightarrow x, x_s^-$ then~$x$ is on side~$s$ or on a side below~$s$;
\item[(c)] if~$x$ and~$y$ are on side~$s$ then~$x\rightarrow y$ if and only if~$x$ is above~$y$.
\end{enumerate}
\end{lemma}}
\begin{proof}
\leo{Part~(b) and~(c) follow directly from the assumption that~$N$ displays~$\cT$. For~(a), assume that~$x$ is on side~$s$ and suppose that $x_s^+ \rightarrow x, x_s^-$ does not hold. Then~$x$ and~$x_s^+$ must have a common parent in all trees in~$\cT$, contradicting the assumption that~$\cT$ has no nontrivial common pendant subtrees.}
\end{proof}

\leo{We are now ready to describe our exponential-time algorithm for \HN, which we do in Algorithm~\ref{alg:alg}.} \revB{The main idea of the algorithm is to guess a partial network and to process its sides bottom-up, such that the remaining leaves on each considered side and their order is determined by the $\rightarrow$ relation.}

\begin{algorithm}[t]
\leo{Apply the subtree reduction (see Algorithm~\ref{alg:kernel}).\\
\For{each partial network~$N_p$ \revB{with respect to}~$X$ with~$r(N_p)\leq k$}{
Let~$G$ be the underlying generator of~$N_p$.\\
Mark each side~$s$ of~$G$ with $|X_{N_p}(s)|\leq 2$ as finished and each other side as unfinished.\\
\While{there exists an unfinished side of~$G$}{
Let~$s$ be an unfinished side of~$G$ such that there is no unfinished side of~$G$ that is below~$s$.\\
Let $x_s^+$ and~$x_s^-$ be, respectively, the top and bottom leaf on side~$s$ in~$N_p$.\\
Let~$X_s$ be the set of all $x\in X$ that are not in~$N_p$ and such that $x_s^+ \rightarrow x, x_s^-$.\\
Let~$x_1,\ldots ,x_{|X_s|}$ be the ordering of the leaves in~$X_s$ such that $x_i\rightarrow x_j$ implies that~$i<j$.\\
Replace the edge between the parent~$v_s^+$ of~$x_s^+$ and the parent~$v_s^-$ of~$x_s^-$ by a directed path $v_s^+,v_1,\ldots ,v_{|X_s|},v_s^-$ and add the leaves in~$X_s$ by edges $(v_i,x_i)$ for~$i=1,\ldots ,|X_s|$.\\
Mark side~$s$ as finished.
}
\If{the obtained network displays~$\cT$}{Output this network.}
}}
\caption{Exponential-time algorithm for \HN.\label{alg:alg}}
\end{algorithm}
\medskip
\leo{\begin{theorem}\label{thm:alg}
There exists an $\revX{n^{f(k)}t}$ time algorithm for \HN, with~$n=|X|$,~$t=|\cT|$ and~$f$ some computable function of~$k$.
\end{theorem}}
\begin{proof}
\leo{We claim that Algorithm~\ref{alg:alg} solves \HN within the claimed running time bound. \revA{Recall that we} may restrict to binary networks by Observation~\ref{obs:binary}. The subtree reduction takes~$O(n^4t)$ time by Lemma~\ref{lem:polysubtree} (since there are at most~$n$ maximal common pendant subtrees) and is safe by Lemma~\ref{lem:subtreereduction}. Hence, we may restrict to networks with no nontrivial pendant subtrees. The number of partial networks~$N_p$ \revB{with respect to}~$X$ with~$r(N_p) \leq k$ is at most \revA{$2^{9k\log{k}+O(k)}(n+1)^{9k}$} by Lemma~\ref{lem:partialnetworks} \cs{(there is an additional factor~$k$ because here we have $r(N_p) \leq k$ instead of  $r(N_p) = k$, \revC{but this factor has been absorbed in the factor $2^{O(k)}$})}. For each partial network~$N_p$, Algorithm~\ref{alg:alg} constructs a network~$N$ on~$X$ displaying~$\cT$ that is consistent with~$N_p$, if such a network exists. Correctness of this construction follows from Lemma~\ref{lem:arrows}. There are, by Lemma~\ref{lem:sides}, at most~$4k-1$ sides that are initially marked as unfinished. Adding leaves to each such side takes $O(n\cdot t)$ time. Finally, to check if the obtained network~$N$ displays~$\cT$, we loop through the set~$\cT(N)$ of the at most~$2^k$ binary trees displayed by~$N$ and check for each such tree if it is a refinement of one or more of the trees in~$\cT$. We have that~$N$ displays~$\cT$ if and only if for each tree~$T\in\cT$ there is at least one tree~$T'\in\cT(N)$ that is a refinement of~$T$. Checking if~$T'$ is a refinement of~$T$ takes at most~$O(n^2)$ time. Therefore, the total running time is $O(n^4t+\revC{2^{9k\log{k}+O(k)}}(n+1)^{9k}((4k-1)n\cdot t+2^kn^2t))$ and hence~$O(\revC{2^{9k\log{k}+O(k)}}(n+1)^{9k}n^2t)$, if~$k\geq 1$.}
\end{proof}

\leo{By combining Theorem~\ref{thm:alg} with Theorems~\ref{thm:kernel} and~\ref{thm:kerneldegree} we obtain the following corollaries.}
\medskip
\begin{corollary}
There exists an \revB{$f(k,t) + \revAnew{O(n^7t)}$} time algorithm for \HN, with~$f$ some computable function of~$k$ and~$t=|\cT|$ and with~$n=|X|$.
\end{corollary}
\medskip
\begin{corollary}
There exists an \revB{$f(k,\Dp) + \revAnew{O(n^7t)}$} time algorithm for \HN, with~$f$ some computable function of~$k$ and the maximum outdegree~$\Dp$ of the input trees and with~$t=|\cT|$ and $n=|X|$.
\end{corollary}

\section{\revX{Experiments}}

\subsection{\revX{Implementation and experimental setup}}

\begin{table}
\footnotesize 
  \centering 
\begin{tabular}{|c|c|c|c|c|c|c|c|c|c|c|}
\hline
$\bar{n}$ & $\bar{r}$ & $\bar{c}$ &  $\bar{s}$ & $\bar{k}$& \begin{minipage}{1.4cm}\centering  subtree  kern.\\  factor \end{minipage} &
\begin{minipage}{1.5cm}\centering  chain kern. \\ (Alg. \ref{alg:kernel}) factor  \end{minipage} &
\begin{minipage}{1.5cm}\centering  chain kern. \\ (Alg. \ref{alg:kernel2})  factor   \end{minipage} &   \begin{minipage}{1.4cm}\centering total kern.\\ factor \end{minipage} &  \begin{minipage}{1.4cm}\centering median\\ RT (s) \end{minipage} &  \begin{minipage}{1.4cm}\centering average\\ RT (s) \end{minipage} \\[20pt]
\hline
500 & 1 & 0 & 50 & 1 & .94 & 0 & .01 & .96 &  3 & 2.7 \\
\hline
500 & 1 & 0 & 50 & 3 & .94 & 0 & 0 & .94 &  3 & 2.7 \\
\hline
500 & 1 & 0 & 98 & 1 & .72 & 0 & .23 & .96 &  100 & 149.4 \\
\hline
500 & 1 & 0 & 98 & 3 & .72 & 0 & .19 & .92 &  79.5 & 96.8 \\
\hline
500 & 1 & 98 & 50 & 1 & .96 & 0 & 0 & .97 &  5 & 5.5 \\
\hline
500 & 1 & 98 & 50 & 3 & .97 & 0 & 0 & .97 &  5 & 5.1 \\
\hline
500 & 1 & 98 & 98 & 1 & .71 & .21 & 0 & .93 &  18 & 24.2 \\
\hline
500 & 1 & 98 & 98 & 3 & .71 & .12 & 0 & .83 &  25.5 & 25.3 \\
\hline
500 & 10 & 0 & 50 & 1 & .72 & 0 & 0 & .73 &  5 & 4.7 \\
\hline
500 & 10 & 0 & 50 & 3 & .72 & 0 & 0 & .72 &  4 & 3.9 \\
\hline
500 & 10 & 0 & 98 & 1 & .48 & 0 & .22 & .71 &  25.5 & 25.6 \\
\hline
500 & 10 & 0 & 98 & 3 & .49 & 0 & .06 & .56 &  23 & 27.2 \\
\hline
500 & 10 & 98 & 50 & 1 & .89 & 0 & 0 & .90 &  6.5 & 6.5 \\
\hline
500 & 10 & 98 & 50 & 3 & .89 & 0 & 0 & .89 &  6.5 & 6 \\
\hline
500 & 10 & 98 & 98 & 1 & .58 & .24 & 0 & .83 &  25 & 24.8 \\
\hline
500 & 10 & 98 & 98 & 3 & .63 & .10 & 0 & .73 &  23 & 31.4 \\
\hline
1000 & 1 & 0 & 50 & 1 & .97 & 0 & .01 & .98 &  5 & 5.1 \\
\hline
1000 & 1 & 0 & 50 & 3 & .96 & 0 & 0 & .96 &  5 & 5.4 \\
\hline
1000 & 1 & 0 & 98 & 1 & .86 & 0 & .11 & .98 &  59.5 & 165.9 \\
\hline
1000 & 1 & 0 & 98 & 3 & .88 & 0 & .07 & .96 &  69 & 119.8 \\
\hline
1000 & 1 & 98 & 50 & 1 & .98 & 0 & 0 & .98 &  37.5 & 37.3 \\
\hline
1000 & 1 & 98 & 50 & 3 & .98 & 0 & 0 & .98 &  35 & 33.4 \\
\hline
1000 & 1 & 98 & 98 & 1 & .81 & .13 & 0 & .95 &  60.5 & 66.6 \\
\hline
1000 & 1 & 98 & 98 & 3 & .84 & .07 & 0 & .92 &  56.5 & 63.5 \\
\hline
1000 & 10 & 0 & 50 & 1 & .83 & 0 & .01 & .84 &  13.5 & 13.1 \\
\hline
1000 & 10 & 0 & 50 & 3 & .83 & 0 & 0 & .83 &  12.5 & 13.1 \\
\hline
1000 & 10 & 0 & 98 & 1 & .65 & 0 & .18 & .83 &  114 & 117 \\
\hline
1000 & 10 & 0 & 98 & 3 & .65 & 0 & .10 & .75 &  178.5 & 179.4 \\
\hline
1000 & 10 & 98 & 50 & 1 & .94 & 0 & 0 & .95 &  58 & 56.1 \\
\hline
1000 & 10 & 98 & 50 & 3 & .94 & 0 & 0 & .94 &  44 & 40.6 \\
\hline
1000 & 10 & 98 & 98 & 1 & .73 & .16 & 0 & .89 &  87 & 103.9 \\
\hline
1000 & 10 & 98 & 98 & 3 & .71 & .08 & 0 & .79 &  103 & 105.5 \\
\hline
\end{tabular}
  \caption{\revBnew{
Kernelization factors and running times for several combinations of parameters $\bar{r}$, $\bar{c}$,  $\bar{s}$, $\bar{k}$ and $\bar{n}$, for~$\bar{t}=3$ trees.} }\label{table:resSim}
\end{table}

\begin{table}
\footnotesize 
  \centering 
\begin{tabular}{|c|c|c|c|c|c|c|c|c|c|c|}
\hline
$\bar{n}$ & $\bar{r}$ & $\bar{c}$ &  $\bar{s}$ & $\bar{k}$& \begin{minipage}{1.4cm}\centering  subtree  kern.\\  factor \end{minipage} &
\begin{minipage}{1.5cm}\centering  chain kern. \\ (Alg. \ref{alg:kernel}) factor  \end{minipage} &
\begin{minipage}{1.5cm}\centering  chain kern. \\ (Alg. \ref{alg:kernel2})  factor   \end{minipage} &   \begin{minipage}{1.4cm}\centering total kern.\\ factor \end{minipage} &  \begin{minipage}{1.4cm}\centering median\\ RT (s) \end{minipage} &  \begin{minipage}{1.4cm}\centering average\\ RT (s) \end{minipage} \\[20pt]
\hline
500 & 1 & 0 & 50 & 1 & 0.93 & 0 & 0.01 & 0.94 & 3 &  3.3 \\
\hline
500 & 1 & 0 & 50 & 3 & 0.92 & 0 & 0 & 0.92 & 3 &  3.3 \\
\hline
500 & 1 & 0 & 98 & 1 & 0.63 & 0 & 0.31 & 0.94 & 148 &  180.7 \\
\hline
500 & 1 & 0 & 98 & 3 & 0.64 & 0 & 0.23 & 0.87 & 213 &  191.3 \\
\hline
500 & 1 & 98 & 50 & 1 & 0.96 & 0 & 0 & 0.97 & 6 &  5.6 \\
\hline
500 & 1 & 98 & 50 & 3 & 0.97 & 0 & 0 & 0.97 & 6 &  5.7 \\
\hline
500 & 1 & 98 & 98 & 1 & 0.7 & 0.2 & 0 & 0.9 & 21.5 &  23.5 \\
\hline
500 & 1 & 98 & 98 & 3 & 0.67 & 0.09 & 0 & 0.76 & 19.5 &  21.2 \\
\hline
500 & 10 & 0 & 50 & 1 & 0.66 & 0 & 0 & 0.67 & 6 &  5.8 \\
\hline
500 & 10 & 0 & 50 & 3 & 0.67 & 0 & 0 & 0.67 & 5 &  5.1 \\
\hline
500 & 10 & 0 & 98 & 1 & 0.45 & 0 & 0.16 & 0.61 & 30 &  28.8 \\
\hline
500 & 10 & 0 & 98 & 3 & 0.44 & 0 & 0.04 & 0.49 & 26.5 &  27.9 \\
\hline
500 & 10 & 98 & 50 & 1 & 0.87 & 0 & 0 & 0.88 & 7.5 &  7.9 \\
\hline
500 & 10 & 98 & 50 & 3 & 0.86 & 0 & 0 & 0.86 & 7 &  7 \\
\hline
500 & 10 & 98 & 98 & 1 & 0.56 & 0.17 & 0 & 0.74 & 29.5 &  30.1 \\
\hline
500 & 10 & 98 & 98 & 3 & 0.59 & 0.05 & 0 & 0.63 & 23.5 &  29 \\
\hline
1000 & 1 & 0 & 50 & 1 & 0.96 & 0 & 0.01 & 0.97 & 6 &  5.9 \\
\hline
1000 & 1 & 0 & 50 & 3 & 0.95 & 0 & 0 & 0.95 & 6 &  6 \\
\hline
1000 & 1 & 0 & 98 & 1 & 0.83 & 0 & 0.13 & 0.96 & 232.5 &  210.4 \\
\hline
1000 & 1 & 0 & 98 & 3 & 0.84 & 0 & 0.1 & 0.94 & 188.5 &  190.8 \\
\hline
1000 & 1 & 98 & 50 & 1 & 0.98 & 0 & 0 & 0.98 & 57.5 &  58 \\
\hline
1000 & 1 & 98 & 50 & 3 & 0.98 & 0 & 0 & 0.98 & 51 &  51.7 \\
\hline
1000 & 1 & 98 & 98 & 1 & 0.79 & 0.13 & 0 & 0.92 & 63 &  65.4 \\
\hline
1000 & 1 & 98 & 98 & 3 & 0.81 & 0.07 & 0 & 0.89 & 65 &  72.3 \\
\hline
1000 & 10 & 0 & 50 & 1 & 0.77 & 0 & 0.01 & 0.78 & 22 &  22.6 \\
\hline
1000 & 10 & 0 & 50 & 3 & 0.77 & 0 & 0 & 0.77 & 20 &  20.1 \\
\hline
1000 & 10 & 0 & 98 & 1 & 0.58 & 0 & 0.17 & 0.76 & 155.5 &  167.2 \\
\hline
1000 & 10 & 0 & 98 & 3 & 0.61 & 0 & 0.05 & 0.66 & 118 &  121 \\
\hline
1000 & 10 & 98 & 50 & 1 & 0.92 & 0 & 0 & 0.93 & 58.5 &  58.6 \\
\hline
1000 & 10 & 98 & 50 & 3 & 0.92 & 0 & 0 & 0.92 & 76.5 &  68.1 \\
\hline
1000 & 10 & 98 & 98 & 1 & 0.67 & 0.16 & 0 & 0.83 & 121 &  137.5 \\
\hline
1000 & 10 & 98 & 98 & 3 & 0.67 & 0.05 & 0 & 0.73 & 116.5 &  126 \\
\hline
\end{tabular}
  \caption{\revBnew{Kernelization factors and running times for several combinations of  parameters $\bar{r}$, $\bar{c}$, $\bar{s}$, $\bar{k}$ and $\bar{n}$, for~$\bar{t}=4$ trees.}}\label{table:resSim2}
\end{table}

To demonstrate the impact of our kernelization algorithms on real instances, we implemented the algorithms presented in Sections \ref{sec:boundedtrees} and \ref{sec:boundedoutdegrees} in Java. The implementation integrates all the reductions into one execution: first it runs Algorithm \ref{alg:kernel2}, then Algorithm \ref{alg:kernel}. We did not implement the exponential-time algorithm presented in Section~\ref{sec:alg} because this is purely a classification result, which is not fast enough for practical use.
We performed some tests on \revBnew{instances with three and four} trees. Given \revBnew{six} parameters \revBnew{($\bar{t}$, $\bar{n}$, $\bar{r}$, $\bar{c}$,  $\bar{s}$ and $\bar{k}$)}, the~\revBnew{$\bar{t}$} trees of each instance are constructed and consequently reduced as follows. First, we generate a random binary tree $T_1$ with \revBnew{$\bar{n}$} taxa and  skew factor $\bar{s}$ (the closer this is to 50, the more balanced the tree is, the closer to 100, the more the tree resembles a chain).  Then \revBnew{each of the other trees is} created from $T_1$ by performing $\bar{r}$ random rSPR moves. Informally an rSPR move is where a subtree is detached and regrafted
elsewhere in the tree: such moves \cite{rspr1,rspr2,rspr3} are often used in experiments to induce increasing hybridization number \cite[among others]{HN}. Finally, in \revBnew{all but one of the trees} a subset of $\bar{c}$ \% of the edges
are contracted. 
The kernelization algorithms are run with $k=\bar{k}$. \revBnew{Table~\ref{table:resSim} shows the results for three trees and Table~\ref{table:resSim2} for four trees. Each row is the average of~10 runs with the
given combination of parameters.} \revBnew{We give the average kernelization factor} -- defined as the ratio between the number of leaves removed and the original number of leaves \revBnew{$\bar{n}$} -- \revBnew{as wel as the average and median running times.} Our aim is to show that our algorithms are practical, by showing that the kernelization factor is high for several combinations of parameters. Moreover, we also want to test for which combinations of parameters the subtree reduction (Algorithm  \ref{alg:kernel})  and chain reductions  (Algorithm  \ref{alg:kernel} and   \ref{alg:kernel2}) are more effective (their kernelization factors are reported in columns 6-8, these are also defined with respect to \revBnew{$\bar{n}$} so can be summed to obtain the total kernelization factor). 

\subsection{\revX{Analysis of experiments}}
\revX{Given the large number of taxa involved (500 and 1000) it is encouraging to observe
that the implementation runs quickly: on a 3.1GHz processor with 4Gb of RAM every
parameter combination terminated within 10 minutes, and many parameter combinations
were significantly faster, \revBnew{see the last two columns of the tables}. Part of the reason for this is the subtree reduction, which is
the asymptotically fastest part of the kernelization. The subtree reduction always
executes first and this has the effect of significantly reducing the number of taxa before the
chain reductions are executed. \revBnew{The kernelization factors of the chain reductions are much smaller, partly because they are calculated relative to the original number of taxa, before the subtree reduction}.

Looking at the table more closely, a number of observations can be made. Clearly, in this experimental setup both chain reductions require that the starting tree $T_1$ is heavily chain-like, which is achieved by having a skew factor close to 100. Otherwise the starting tree $T_1$ is too ``bushy'' and under the action of rSPR moves no long chains are formed. \revBnew{Secondly, if there is no contraction}, then all the trees are binary, and the degree-based chain reduction (Algorithm  \ref{alg:kernel2}) has quite a large impact, while Algorithm \ref{alg:kernel} has no impact at all. If there is an extremely large amount of contraction, then the roles of the two chain reductions are reversed.
\revBnew{An intermediate amount of contraction (not shown in the table) effectively disables both chain reductions, but not the subtree reduction}. \revBnew{Conversely, a growing number of rSPR moves (which have the effect of increasing the topological dissimilarity of the trees) reduces the impact of the subtree reduction but not of the chain reduction.} As $k=\bar{k}$ increases, the impact of both chain reductions diminishes, due to the increasing of  the  length at which chains are truncated. However,
the impact of the reductions \revBnew{decreases only slightly when~$k$ is increased from~1 to~3}. This is encouraging, suggesting that both chain reductions can have an impact for larger values of~$k$. \revBnew{Similarly, increasing the number of trees has only a very small negative effect on the impact of the subtree and chain reductions. Finally, as the \revBnew{total kernelization factors in the table show}, it is clear that the kernelization ``works'': for all \revBnew{parameter combinations in the tables} the instances reduce in size by at least 49$\%$.}

\section{Discussion and open problems}
\label{sec:blahblah}

The main open question remains whether the \HN problem is \leo{fixed-parameter tractable, and if it has a polynomial kernel, when parameterized only by~$k$ (\revB{that is}, when the number of input trees and their outdegrees are unbounded).}

Note that when the input trees are not required to have the same label set~$X$, \HN is not fixed-parameter tractable unless $\textsf{P}$ $=$ $\textsf{NP}$. The reason for this is that it is $\textsf{NP}$-hard to decide if~$r(\cT)=1$ for sets~$\cT$ consisting of \revA{trees} with three leaves each~\cite[Theorem~7]{JanssonEtAl2006}.

Another question is whether the kernel size can be reduced for certain fixed~$|\cT|$. For~$|\cT|=2$, our results give a cubic kernel, while Linz and Semple~\cite{linzsemple2009} showed a linear kernel of a modified, weighted problem, by analyzing carefully how common chains can look in two trees. Can something like this be done for more than two trees? In particular, does there exist a quadratic kernel for three trees? \revB{In addition, can the running times of the kernelization algorithms be reduced?}

Finally, there is the problem of solving the kernelized instances. For this, a fast exponential-time exact algorithm is needed, or a good heuristic. \leo{Although we have presented an $O(n^{f(k)}t)$ time algorithm for \HN, with~$n=|X|$ and~$t=|\cT|$, it is not known if there exists an~$O(c^n)$-algorithm for some constant~$c$. While $O(c^kn^{O(1)})$ algorithms have been developed for instances consisting of two binary trees~\cite{whidden2013fixed} and very recently for three binary trees~\cite{threetrees}, it is not clear if they exist for four or more binary trees, or for two or more nonbinary trees. Note that, for practical applications, the kernelization can also be combined with an efficient heuristic.}

\section*{\revX{Acknowledgements}} \revX{We thank the anonymous reviewers for their helpful comments.}

\bibliographystyle{elsarticle-num-names}

\end{document}